\documentclass[envcountsame, runningheads]{llncs}

\usepackage[utf8]{inputenc}
\usepackage[T1]{fontenc}
\usepackage{lmodern}
\usepackage[hidelinks]{hyperref}
\usepackage{diagbox}

\usepackage{amssymb}
\usepackage{amsmath}
\usepackage{tikz}
\usetikzlibrary{backgrounds, positioning, arrows, calc}
\usepackage{mathtools}
\usepackage[linesnumbered,vlined,noend]{algorithm2e}
\usepackage[inline]{enumitem}  

\usepackage{todonotes}

\tikzstyle{state}=[thick,minimum size=18pt, circle,draw]
\tikzstyle{transition}=[->,thick,>=stealth,shorten >=1pt,shorten <=1pt]
\tikzstyle{final}=[after node path={ node[state, scale=.8] at (\tikzlastnode) {} }]
\tikzstyle{initial}=[after node path={
	[to path={[transition] (\tikztostart) -- (\tikztotarget)}]
	(\tikzlastnode)++(180:22pt) edge (\tikzlastnode)
}]
\tikzset{
	bg/.default={},
	bg/.style={execute at end picture={
			\begin{scope}[on background layer]
				\node[xshift=-1mm, yshift=-1mm] (sw) at (current bounding box.south west) {};
				\node[xshift=1mm, yshift=1mm] (ne) at (current bounding box.north east) {};
				\node[xshift=1mm, yshift=-1mm] (nw) at (current bounding box.north west) {};
				\fill[fill=black!10,rounded corners] (sw) rectangle (ne);
				
				\ifx&#1&\else
				\node[anchor=north east, xshift=2pt] at (nw) {#1};
				\fi
			\end{scope}
	}},
}

\newcommand{\N}{\mathbb{N}}

\newcommand{\R}{\mathbb{R}}
\newcommand{\RR}{\R^{\pm \infty}}

\newcommand{\D}{\mathbb{D}}
\newcommand{\B}{\mathbb{B}}

\newcommand{\avg}{\text{\normalfont avg}}

\newcommand{\fin}{\text{\normalfont fin}}

\newcommand{\LTLf}{\ensuremath{\lozenge}}

\newcommand{\LTLg}{\ensuremath{\square}}

\let\LTLnext\LTLo

\newcommand{\req}{\texttt{rq}}

\newcommand{\gra}{\texttt{gr}}
\newcommand{\tick}{\texttt{tk}}
\newcommand{\other}{\texttt{oo}}

\newcommand{\calM}{\mathcal{M}}

\newcommand{\prefixeq}{\preceq}
\newcommand{\prefix}{\prec}

\newcommand{\suchthat}{\;\ifnum\currentgrouptype=16 \middle\fi|\;}
\let\st\suchthat

\title{Quantitative Safety and Liveness} 
\author{Thomas~A.~Henzinger \and Nicolas~Mazzocchi \and N.~Ege~Sara\c{c}} %
\institute{Institute of Science and Technology Austria (ISTA), Klosterneuburg, Austria \\ \email{\{tah,nmazzocc,esarac\}@ist.ac.at}} %
\authorrunning{T.A.~Henzinger et al.} %
\date{}

\begin{document}
	\maketitle
	\begin{abstract}
		Safety and liveness are elementary concepts of computation, and the foundation of many verification paradigms.
		The safety-liveness classification of boolean properties characterizes whether a given property can be falsified by observing a finite prefix of an infinite computation trace (always for safety, never for liveness).
		In quantitative specification and verification, properties assign not truth values, but quantitative values to infinite traces (e.g., a cost, or the distance to a boolean property). 
		We introduce quantitative safety and liveness, and we prove that our definitions induce conservative quantitative generalizations of both (1)~the safety-progress hierarchy of boolean properties and (2)~the safety-liveness decomposition of boolean properties.
		In particular, we show that every quantitative property can be written as the pointwise minimum of a quantitative safety property and a quantitative liveness property.
		Consequently, like boolean properties, also quantitative properties can be $\min$-decomposed into safety and liveness parts, or alternatively, $\max$-decomposed into co-safety and co-liveness parts.
		Moreover, quantitative properties can be approximated naturally. 
		We prove that every quantitative property that has both safe and co-safe approximations can be monitored arbitrarily precisely by a monitor that uses only a finite number of states. 
	\end{abstract}

\section{Introduction}        
Safety and liveness are elementary concepts in the semantics of computation~\cite{DBLP:journals/tse/Lamport77}.
They can be explained through the thought experiment of a \emph{ghost monitor}---an imaginary device that watches an infinite computation trace at runtime, one observation at a time, and always maintains the set of \emph{possible prediction values} to reflect the satisfaction of a given property.
Let $\varPhi$ be a boolean property, meaning that $\varPhi$ divides all infinite traces into those that satisfy~$\varPhi$, and those that violate~$\varPhi$.
After any finite number of observations,~\texttt{True} is a possible prediction value for $\varPhi$ if the observations seen so far are consistent with an infinite trace that satisfies~$\varPhi$, and~\texttt{False} is a possible prediction value for $\varPhi$ if the observations seen so far are consistent with an infinite trace that violates~$\varPhi$.
When~\texttt{True} is no possible prediction value, the ghost monitor can reject the hypothesis that $\varPhi$ is satisfied.
The property $\varPhi$ is \emph{safe} if and only if the ghost monitor can always reject the hypothesis $\varPhi$ after a finite number of observations:
if the infinite trace that is being monitored violates~$\varPhi$, then after some finite number of observations, \texttt{True} is no possible prediction value for~$\varPhi$.
Orthogonally, the property $\varPhi$ is \emph{live} if and only if the ghost monitor can never reject the hypothesis $\varPhi$ after a finite number of observations:
for all infinite traces, after every finite number of observations,~\texttt{True} remains a possible prediction value for~$\varPhi$.

The safety-liveness classification of properties is fundamental in verification.
In the natural topology on infinite traces---the ``Cantor topology''---the safety properties are the closed sets, and the liveness properties are the dense sets~\cite{DBLP:journals/ipl/AlpernS85}.
For every property~$\varPhi$, the location of $\varPhi$ within the Borel hierarchy that is induced by the Cantor topology---the so-called ``safety-progress hierarchy''~\cite{ChangMP93}---indicates the level of difficulty encountered when verifying~$\varPhi$.
On the first level, we find the safety and co-safety properties, the latter being the complements of safety properties, i.e., the properties whose falsehood (rather than truth) can always be rejected after a finite number of observations by the ghost monitor.
More sophisticated verification techniques are needed for second-level properties, which are the countable boolean combinations of first-level properties---the so-called ``response'' and ``persistence'' properties~\cite{ChangMP93}.
Moreover, the orthogonality of safety and liveness leads to the following celebrated fact:
\emph{every} property can be written as the intersection of a safety property and a liveness property~\cite{DBLP:journals/ipl/AlpernS85}.
This means that every property $\varPhi$ can be decomposed into two parts:
a safety part---which is amenable to simple verification techniques, such as invariants---and a liveness part---which requires heavier verification paradigms, such as ranking functions.
Dually, there is always a disjunctive decomposition of $\varPhi$ into co-safety and co-liveness.

So far, we have retold the well-known story of safety and liveness for \emph{boolean} properties.
A boolean property $\varPhi$ is formalized mathematically as the \emph{set} of infinite computation traces that satisfy~$\varPhi$, or equivalently, the characteristic \emph{function} that maps each infinite trace to a truth value.
Quantitative generalizations of the boolean setting allow us to capture not only correctness properties, but also performance properties~\cite{DBLP:conf/concur/HenzingerO13}.
In this paper we reveal the story of safety and liveness for such \emph{quantitative} properties, which are functions from infinite traces to an arbitrary set $\D$ of \emph{values}.
In order to compare values, we equip the value domain $\D$ with a partial order~$<$, and we require $(\D,<)$ to be a complete lattice.
The membership problem~\cite{DBLP:journals/tocl/ChatterjeeDH10} for an infinite trace $f$ and a quantitative property $\varPhi$ asks whether $\varPhi(f)\geq v$ for a given threshold value $v\in\D$.
Correspondingly, in our thought experiment, the ghost monitor attempts to reject hypotheses of the form $\varPhi(f)\geq v$, which cannot be rejected as long as all observations seen so far are consistent with an infinite trace $f$ with $\varPhi(f)\geq v$.
We will define $\varPhi$ to be a \emph{quantitative safety} property if and only if every hypothesis of the form $\varPhi(f)\geq v$ can always be rejected by the ghost monitor after a finite number of observations, and we will define $\varPhi$ to be a \emph{quantitative liveness} property if and only if some hypothesis of the form $\varPhi(f)\geq v$ can never be rejected by the ghost monitor after any finite number of observations.
We note that in the quantitative case, after every finite number of observations, the set of possible prediction values for $\varPhi$ maintained by the ghost monitor may be finite or infinite, and in the latter case, it may not contain a minimal or maximal element.

Let us give a few examples.
Suppose we have four observations:
observation $\req$ for ``request a resource,'' observation $\gra$ for ``grant the resource,'' observation $\tick$ for ``clock tick,'' and observation $\other$ for ``other.''
The boolean property \textsf{Resp} requires that every occurrence of~$\req$ in an infinite trace is followed eventually by an occurrence of~$\gra$.
The boolean property \textsf{NoDoubleReq} requires that no occurrence of~$\req$ is followed by another~$\req$ without some~$\gra$ in between.
The quantitative property \textsf{MinRespTime} maps every infinite trace to the largest number $k$ such that there are at least $k$ occurrences of~$\tick$ between each~$\req$ and the closest subsequent~$\gra$.
The quantitative property \textsf{MaxRespTime} maps every infinite trace to the smallest number $k$ such that there are at most $k$ occurrences of~$\tick$ between each~$\req$ and the closest subsequent~$\gra$.
The quantitative property \textsf{AvgRespTime} maps every infinite trace to the lower limit value $\liminf$ of the infinite sequence $(v_i)_{i \geq 1}$, where $v_i$ is, for the first $i$ occurrences of~$\tick$, the average number of occurrences of~$\tick$ between~$\req$ and the closest subsequent~$\gra$.
Note that the values of \textsf{AvgRespTime} can be $\infty$ for some computations, including those for which the value of \textsf{Resp} is~\texttt{True}.
This highlights that boolean properties are not embedded in the limit behavior of quantitative properties.

The boolean property \textsf{Resp} is live because every finite observation sequence can be extended with an occurrence of~$\gra$.
In fact, \textsf{Resp} is a second-level liveness property (namely, a response property), because it can be written as a countable intersection of co-safety properties. 
The boolean property \textsf{NoDoubleReq} is safe because if it is violated, it will be rejected by the ghost monitor after a finite number of observations, namely, as soon as the ghost monitor sees a~$\req$ followed by another occurrence of~$\req$ without an intervening~$\gra$.
According to our quantitative generalization of safety, \textsf{MinRespTime} is a safety property.
The ghost monitor always maintains the minimal number $k$ of occurrences of~$\tick$ between any past~$\req$ and the closest subsequent~$\gra$ seen so far; the set of possible prediction values for \textsf{MinRespTime} is always $\{0,1,\ldots,k\}$.
Every hypothesis of the form ``the \textsf{MinRespTime}-value is at least~$v$'' is rejected by the ghost monitor as soon as $k<v$; if such a hypothesis is violated, this will happen after some finite number of observations.
Symmetrically, the quantitative property \textsf{MaxRespTime} is co-safe, because every wrong hypothesis of the form ``the \textsf{MaxRespTime}-value is at most~$v$'' will be rejected by the ghost monitor as soon as
the smallest possible prediction value for \textsf{MaxRespTime}, which is the maximal number of occurrences of~$\tick$ between any past~$\req$ and the closest subsequent~$\gra$ seen so far, goes above~$v$.
By contrast, the quantitative property \textsf{AvgRespTime} is both live and co-live because no hypothesis of the form ``the \textsf{AvgRespTime}-value is at least~$v$,'' nor of the form ``the \textsf{AvgRespTime}-value is at most~$v$,'' can ever be rejected by the ghost monitor after a finite number of observations.
All nonnegative real numbers and $\infty$ always remain possible prediction values for \textsf{AvgRespTime}.
Note that a ghost monitor that attempts to reject hypotheses of the form $\varPhi(f) \geq v$ does not need to maintain the entire set of possible prediction values, but only the $\sup$ of the set of possible prediction values, and whether or not the $\sup$ is contained in the set.
Dually, updating $\inf$ (and whether it is contained) suffices to reject hypotheses of the form $\varPhi(f) \leq v$.

By defining quantitative safety and liveness via ghost monitors, we not only obtain a conservative and quantitative generalization of the boolean story, but also open up attractive frontiers for quantitative semantics, monitoring, and verification.
For example, while the approximation of boolean properties reduces to adding and removing traces to and from a set, the approximation of quantitative properties offers a rich landscape of possibilities.
In fact, we can approximate the notion of safety itself.
Given an error bound~$\alpha$, the quantitative property $\varPhi$ is \emph{$\alpha$-safe} if and only if for every value $v$ and every infinite trace $f$ whose value $\varPhi(f)$ is less than~$v$, all possible prediction values for $\varPhi$ are less than $v+\alpha$ after some finite prefix of~$f$.
This means that for an $\alpha$-safe property~$\varPhi$, the ghost monitor may not reject wrong hypotheses of the form $\varPhi(f)\ge v$ after a finite number of observations, once the violation is below the error bound.
We show that every quantitative property that is both $\alpha$-safe and $\beta$-co-safe, for any finite~$\alpha$~and~$\beta$, can be monitored arbitrarily precisely by a monitor that uses only a finite number of states. 

We are not the first to define quantitative (or multi-valued) definitions of safety and liveness~\cite{DBLP:journals/isci/LiDL17,DBLP:conf/nfm/GorostiagaS22}.
While the previously proposed quantitative generalizations of safety share strong similarities with our definition (without coinciding completely), our quantitative generalization of liveness is entirely new.
The definitions of~\cite{DBLP:conf/nfm/GorostiagaS22} do not support any safety-liveness decomposition, because their notion of safety is too permissive, and their liveness too restrictive.
While the definitions of~\cite{DBLP:journals/isci/LiDL17} admit a safety-liveness decomposition, our definition of liveness captures strictly fewer properties.
Consequently, our definitions offer a stronger safety-liveness decomposition theorem.
Our definitions also fit naturally with the definitions of emptiness, equivalence, and inclusion for quantitative languages~\cite{DBLP:journals/tocl/ChatterjeeDH10}.

\paragraph*{\bf Overview.}
In Section~2, we introduce quantitative properties.
In Section~3, we define quantitative safety as well as safety closure, namely, the property that increases the value of each trace as little as possible to achieve safety.
Then, we prove that our definitions preserve classical boolean facts.
In particular, we show that a quantitative property $\varPhi$ is safe if and only if $\varPhi$ equals its safety closure if and only if $\varPhi$ is upper semicontinuous.
In Section~4, we generalize the safety-progress hierarchy to quantitative properties.
We first define limit properties.
For $\ell \in\{\inf, \sup, \liminf, \limsup\}$, the class of $\ell$-properties captures those for which the value of each infinite trace can be derived by applying the limit function $\ell$ to the infinite sequence of values of finite prefixes.
We prove that $\inf$-properties coincide with safety, $\sup$-properties with co-safety, $\liminf$-properties are suprema of countably many safety properties, and $\limsup$-properties infima of countably many co-safety properties.
The $\liminf$-properties generalize the boolean persistence properties of~\cite{ChangMP93}; the $\limsup$-properties generalize their response properties.
For example, \textsf{AvgRespTime} is a $\liminf$-property.
In Section~5, we introduce quantitative liveness and co-liveness.
We prove that our definitions preserve the classical boolean facts, and show that there is a unique property which is both safe and live.
As main result, we provide a safety-liveness decomposition that holds for every quantitative property.
In Section~6, we define approximate safety and co-safety.
We generalize the well-known unfolding approximation of discounted properties for approximate safety and co-safety properties over the extended reals.
This allows us to provide a finite-state approximate monitor for these properties.
In Section~7, we conclude with future research directions.

\paragraph*{\bf Related Work.}
The notions of safety and liveness for boolean properties appeared first in~\cite{DBLP:journals/tse/Lamport77} and were later formalized in~\cite{DBLP:journals/ipl/AlpernS85}, where safety properties were characterized as closed sets of the Cantor topology on infinite traces, and liveness properties as dense sets.
As a consequence, the seminal decomposition theorem followed:
every boolean property is an intersection of a safety property and a liveness property.
A benefit of such a decomposition lies in the difference between the mathematical arguments used in their verification.
While safety properties enable simpler methods such as invariants, liveness properties require more complex approaches such as well-foundedness~\cite{DBLP:journals/scp/MannaP84,DBLP:journals/dc/AlpernS87}.
These classes were characterized in terms of B\"{u}chi automata in~\cite{DBLP:journals/dc/AlpernS87} and in terms of linear temporal logic in~\cite{DBLP:journals/fac/Sistla94}.

The safety-progress classification of boolean properties~\cite{ChangMP93} proposes an orthogonal view:
rather than partitioning the set of properties, it provides a hierarchy of properties starting from safety.
This yields a more fine-grained view of nonsafety properties which distinguishes whether a ``good thing'' happens at least once (co-safety or ``guarantee''), infinitely many times (response), or eventually always (persistence).
This classification follows the Borel hierarchy that is induced by the Cantor topology on infinite traces, and has corresponding projections within properties that are definable by finite automata and by formulas of linear temporal logic.

Runtime verification, or monitoring, is a lightweight, dynamic verification technique~\cite{DBLP:series/lncs/BartocciFFR18}, where a monitor watches a system during its execution and tries to decide, after each finite sequence of observations, whether the observed finite computation trace or its unknown infinite extension satisfies a desired property.
The safety-liveness dichotomy has profound implications for runtime verification as well:
safety is easy to monitor~\cite{DBLP:conf/tacas/HavelundR02}, while liveness is not.
An early definition of boolean monitorability was equivalent to safety with recursively enumerable sets of bad prefixes~\cite{DBLP:journals/entcs/KimKLSV02}.
The monitoring of infinite-state boolean safety properties was later studied in~\cite{DBLP:conf/lics/FerrereHS18}.
A more popular definition of boolean monitorability~\cite{DBLP:conf/fm/PnueliZ06,DBLP:journals/tosem/BauerLS11} accounts for both truth and falsehood, establishing the set of monitorable properties as a strict superset of finite boolean combinations of safety and co-safety~\cite{DBLP:journals/sttt/FalconeFM12}.
Boolean monitors that use the set possible prediction values can be found in~\cite{DBLP:journals/logcom/BauerLS10}. 
The notion of boolean monitorability was investigated through the safety-liveness lens in~\cite{DBLP:conf/birthday/PeledH18} and through the safety-progress lens in~\cite{DBLP:journals/sttt/FalconeFM12}.

Quantitative properties (a.k.a.~``quantitative languages'')~\cite{DBLP:journals/tocl/ChatterjeeDH10} extend their boolean counterparts by moving from the two-valued truth domain to richer domains such as real numbers.
Such properties have been extensively studied from a static verification perspective in the past decade, e.g., in the context of
games with quantitative objectives~\cite{DBLP:reference/mc/BloemCJ18,DBLP:journals/acta/BouyerMRLL18}, specifying quantitative properties~\cite{DBLP:journals/tocl/BokerCHK14,DBLP:journals/tcs/AlfaroFHMS05}, measuring distances between systems~\cite{DBLP:conf/icalp/AlfaroFS04,DBLP:journals/tcs/CernyHR12,DBLP:journals/tcs/FahrenbergL14,DBLP:journals/ife/Henzinger13}, best-effort synthesis and repair~\cite{DBLP:conf/cav/BloemCHJ09,DBLP:conf/cav/DAntoniSS16}, and quantitative analysis of transition systems~\cite{DBLP:journals/jlp/ThraneFL10,DBLP:journals/cacm/BouyerFLM11,DBLP:conf/aplas/FahrenbergL13,DBLP:journals/tocl/ChatterjeeHO17}.
More recently, quantitative properties have been also studied from a runtime verification perspective, e.g., for limit monitoring of statistical indicators of infinite traces~\cite{DBLP:conf/csl/FerrereHK20} and for analyzing resource-precision trade-offs in the design of quantitative monitors~\cite{DBLP:conf/lics/HenzingerS21,DBLP:conf/rv/HenzingerMS22}.

To the best of our knowledge, previous definitions of (approximate) safety and liveness in nonboolean domains make implicit assumptions about the specification language~\cite{DBLP:conf/atva/WeinerHKPS13,DBLP:conf/csl/KatoenSZ14,DBLP:journals/acta/FaranK18,DBLP:journals/tr/QianSCP22}.
We identify two notable exceptions.
In~\cite{DBLP:conf/nfm/GorostiagaS22}, the authors generalize the framework of~\cite{DBLP:conf/birthday/PeledH18} to nonboolean value domains.
They provide neither a safety-liveness decomposition of quantitative properties, nor a fine-grained classification of nonsafety properties.
In~\cite{DBLP:journals/isci/LiDL17}, the authors present a safety-liveness decomposition and some levels of the safety-progress hierarchy on multi-valued truth domains, which are bounded distributive lattices.
Their motivation is to provide algorithms for model-checking properties on multi-valued truth domains.
We present the relationships between their definitions and ours in the relevant sections below.

\section{Quantitative Properties}
Let $\Sigma = \{a,b,\ldots\}$ be a finite alphabet of observations.
A \emph{trace} is an infinite sequence of observations, denoted by $f,g,h \in \Sigma^\omega$, and a \emph{finite trace} is a finite sequence of observations, denoted by $s,r,t \in \Sigma^*$.
Given $s \in \Sigma^*$ and $w \in \Sigma^* \cup \Sigma^\omega$, we denote by $s \prefix w$ (resp.\ $s \prefixeq w$) that $s$ is a strict (resp.\ nonstrict) prefix of $w$.
Furthermore, we denote by $|w|$ the length of $w$ and, given $a \in \Sigma$, by $|w|_a$ the number of occurrences of $a$ in $w$.

A \emph{value domain} $\D$ is a poset.
Unless otherwise stated, we assume that $\D$ is a nontrivial (i.e., $\bot \neq \top$) complete lattice and, whenever appropriate, we write $0, 1, -\infty, \infty$ instead of $\bot$ and $\top$ for the least and the greatest elements.
We respectively use the terms minimum and maximum for the greatest lower bound and the least upper bound of finitely many elements.

\begin{definition}[Property]
	A \emph{quantitative property} (or simply \emph{property}) is a function $\varPhi : \Sigma^\omega \to \D$ from the set of all traces to a value domain.
\end{definition}

A boolean property $P \subseteq \Sigma^\omega$ is defined as a set of traces.
We use the boolean domain $\B = \{0,1\}$ with $0 < 1$ and, in place of $P$, its \emph{characteristic property} $\varPhi_P : \Sigma^\omega \to \B$, which is defined by $\varPhi_P(f) = 1$ if $f \in P$, and $\varPhi_P(f) = 0$ if $f \notin P$.

For all properties $\varPhi_1,\varPhi_2$ on a domain $\D$ and all traces $f \in \Sigma^\omega$, we let $\min(\varPhi_1,\varPhi_2)(f) = \min(\varPhi_1(f),\varPhi_2(f))$ and $\max(\varPhi_1,\varPhi_2)(f) = \max(\varPhi_1(f),\varPhi_2(f))$.
For a domain $\D$, the \emph{inverse} of $\D$ is the domain $\overline{\D}$ that contains the same elements as $\D$ but with the ordering reversed.
For a property $\varPhi$, we define its \emph{complement} $\overline{\varPhi} : \Sigma^\omega \to \overline{\D}$ by $\overline{\varPhi}(f) = \varPhi(f)$ for all $f \in \Sigma^\omega$.

Some properties can be defined as limits of value sequences. 
A \emph{finitary property} $\pi \colon \Sigma^* \rightarrow \D$ associates a value with each finite trace.
A \emph{value function} $\ell \colon \D^\omega\rightarrow\D$ condenses an infinite sequence of values to a single value.
Given a finitary property $\pi$, a value function $\ell$, and a trace $f \in \Sigma^\omega$, we write $\ell_{s \prefix f} \pi(s)$ instead of $\ell(\pi(s_0)\pi(s_1)\ldots)$, where each $s_i$ fulfills $s_i \prefix f$ and $|s_i|=i$.

\section{Quantitative Safety}
Given a property $\varPhi : \Sigma^\omega \to \D$, a trace $f \in \Sigma^\omega$, and a value $v \in \D$, the quantitative membership problem~\cite{DBLP:journals/tocl/ChatterjeeDH10} asks whether $\varPhi(f) \geq v$.
We define quantitative safety as follows: the property $\varPhi$ is safe iff every wrong hypothesis of the form $\varPhi(f) \geq v$ has a finite witness $s \prec f$.

\begin{definition}[Safety]
	A property $\varPhi : \Sigma^\omega \to \D$ is \emph{safe} iff for every $f \in \Sigma^\omega$ and value $v \in \D$ with $\varPhi(f) \not \geq v$, there is a prefix $s \prefix f$ such that $\sup_{g \in \Sigma^\omega} \varPhi(sg) \not \geq v$.
\end{definition}

Let us illustrate this definition with the \emph{minimal response-time} property.

\begin{example} \label{ex:minresp}
	Let $\Sigma = \{\req, \gra, \tick, \other\}$ and $\D = \N \cup \{\infty\}$.
	We define the minimal response-time property $\varPhi_{\min}$ through an auxiliary finitary property $\pi_{\min}$ that computes the minimum response time so far.
	In a finite or infinite trace, an occurrence of $\req$ is \emph{granted} if it is followed, later, by a $\gra$, and otherwise it is \emph{pending}.
	Let $\pi_{\text{last}}(s) = \infty$ if the finite trace $s$ contains a pending $\req$, or no $\req$, and $\pi_{\text{last}}(s) = |r|_\tick - |t|_\tick$ otherwise, where $r \prec s$ is the longest prefix of $s$ with a pending $\req$, and $t \prec r$ is the longest prefix of $r$ without pending $\req$.
	Intuitively, $\pi_{\text{last}}$ provides the response time for the last request when all requests are granted, and $\infty$ when there is a pending request or no request.
	Given $s \in \Sigma^*$, taking the minimum of the values of $\pi_{\text{last}}$ over the prefixes $r \preceq s$ gives us the minimum response time so far.
	Let $\pi_{\min}(s) = \min_{r \preceq s} \pi_{\text{last}}(r)$ for all $s \in \Sigma^*$, and $\varPhi_{\min}(f) = \lim_{s \prec f} \pi_{\min}(s)$ for all $f\in\Sigma^\omega$.
	The limit always exists because the minimum is monotonically decreasing.
	
	The minimal response-time property is safe.
	Let $f \in \Sigma^\omega$ and $v \in \D$ such that $\varPhi_{\min}(f) < v$.
	Then, some prefix $s \prec f$ contains a $\req$ that is granted after $u < v$ ticks, in which case, no matter what happens in the future, the minimal response time is guaranteed to be at most $u$; that is, $\sup_{g \in \Sigma^\omega} \varPhi_{\min}(sg) \leq u < v$.
	If you recall from the introduction the ghost monitor that maintains the $\sup$ of possible prediction values for the minimal response-time property, that value is always $\pi_{\min}$; that is, $\sup_{g \in \Sigma^\omega} \varPhi_{\min}(sg) = \pi_{\min}(s)$ for all $s\in \Sigma^*$.
	Note that in the case of minimal response time, the $\sup$ of possible prediction values is always realizable; that is, for all $s\in \Sigma^*$, there exists an $f\in\Sigma^{\omega}$ such that $\sup_{g \in \Sigma^\omega} \varPhi_{\min}(sg) = \varPhi_{\min}(sf)$.
\qed\end{example}

\begin{proposition} \label{prop:safetyboolean}
	Quantitative safety generalizes boolean safety.
	For every boolean property $P \subseteq \Sigma^\omega$, the following statements are equivalent:
		(i)~$P$ is safe according to the classical definition~\cite{DBLP:journals/ipl/AlpernS85},
        (ii)~its characteristic property $\varPhi_P$ is safe, and
        (iii)~for every $f \in \Sigma^\omega$ and $v \in \B$ with $\varPhi_P(f) < v$, there exists a prefix $s \prefix f$ such that for all $g \in \Sigma^\omega$, we have $\varPhi_P(sg) < v$.
\end{proposition}
\begin{proof}
	Recall that (i) means the following: for every $f \notin P$ there exists $s \prec f$ such that for all $g \in \Sigma^\omega$ we have $sg \notin P$.
	Expressing the same statement with the characteristic property $\varPhi_P$ of $P$ gives us for every $f \in \Sigma^\omega$ with $\varPhi_P(f) = 0$ there exists $s \prec f$ such that for all $g \in \Sigma^\omega$ we have $\varPhi_P(sg) = 0$.
	In particular, since $\B = \{0,1\}$ and $0 < 1$, we have for every $f \in \Sigma^\omega$ with $\varPhi_P(f) < 1$ there exists $s \prec f$ such that for all $g \in \Sigma^\omega$ we have $\varPhi_P(sg) < 1$.
	Moreover, since there is no $f \in \Sigma^\omega$ with $\varPhi_P(f) < 0$, we get the equivalence between (i) and (iii).
	Now, observe that for every $s \in \Sigma^*$, we have $\varPhi_P(sg) < 1$ for all $g \in \Sigma^\omega$ iff $\sup_{g \in \Sigma^\omega} \varPhi_P(sg) < 1$, simply because the domain $\B$ is a finite total order.
	Therefore, (ii) and (iii) are equivalent as well.
\qed\end{proof}

We now generalize the notion of safety closure and present an operation that makes a property safe by increasing the value of each trace as little as possible.

\begin{definition}[Safety closure]
  The \emph{safety closure} of a property $\varPhi$ is the property $\varPhi^*$ defined by $\varPhi^*(f) = \inf_{s \prec f} \sup_{g \in \Sigma^\omega} \varPhi(sg)$ for all $f\in\Sigma^\omega$. 
\end{definition}

We can say the following about the safety closure operation.

\begin{proposition}\label{proposition:safe:closure}
For every property $\varPhi : \Sigma^\omega \to \D$, the following statements hold.
\begin{enumerate}
	\item $\varPhi^*$ is safe.
	\item $\varPhi^*(f) \geq \varPhi(f)$ for all $f\in\Sigma^\omega$.
	\item $\varPhi^*(f) = {\varPhi^*}^*(f)$ for all $f\in\Sigma^\omega$.
	\item For every safety property $\varPsi : \Sigma^\omega \to \D$, if $\varPhi(f) \leq \varPsi(f)$ for all $f\in\Sigma^\omega$, then $\varPsi(g) \not< \varPhi^*(g)$ for all $g\in\Sigma^\omega$.
\end{enumerate}
\end{proposition}
\begin{proof}
	We first prove that $\sup_{g \in \Sigma^\omega} \varPhi^*(sg) \leq \sup_{g \in \Sigma^\omega} \varPhi(sg)$ for all $s\in\Sigma^*$, in other words, $\sup_{g \in \Sigma^\omega} \inf_{r \prefix sg} \sup_{h \in \Sigma^\omega} \varPhi(rh) \leq \sup_{g \in \Sigma^\omega} \varPhi(sg)$ for all $s\in\Sigma^*$.
	This will be useful for the proofs of the first and the third items above.
	$$\begin{array}{l}
		\forall s, \sup_{g \in \Sigma^\omega} \varPhi(sg) \in \{ \sup_{h \in \Sigma^\omega} \varPhi(rh) \st r \prefixeq s \}
		\\
		\implies \forall s, \sup_{g \in \Sigma^\omega} \varPhi(sg) \geq \inf_{r \prefixeq s} \sup_{h \in \Sigma^\omega} \varPhi(rh)
		\\
		\implies \forall s, \sup_{g \in \Sigma^\omega} \varPhi(sg) \geq  \sup_{g\in\Sigma^\omega} \inf_{r \prefixeq s} \sup_{h \in \Sigma^\omega} \varPhi(rh) \hfill (\dagger)
		\medskip\\
		\forall s, s',  \sup_{g \in \Sigma^\omega} \varPhi(sg) \geq \sup_{h \in \Sigma^\omega} \varPhi(ss'h)
		\\
		\implies \forall s, \sup_{g \in \Sigma^\omega} \varPhi(sg) \geq  \sup_{g\in\Sigma^\omega} \inf_{s' \prefix g} \sup_{h \in \Sigma^\omega} \varPhi(ss'h) \hfill(\ddagger)
		\medskip\\
		(\dagger) \land (\ddagger) \implies  \forall s, \sup_{g \in \Sigma^\omega} \varPhi(sg) \geq  \sup_{g\in\Sigma^\omega} \inf_{r \prefix sg} \sup_{h \in \Sigma^\omega} \varPhi(rh)
	\end{array}$$
	
	Now, we prove that $\varPhi^*$ is safe.
	Suppose towards contradiction that $\varPhi^*$ is not safe, i.e., there exists $f$ and $v$ for which $\varPhi^*(f) \ngeq v$ and $\sup_{g \in \Sigma^\omega} \varPhi^*(sg) \geq v$ for all $s \prefix f$.	
	As a direct consequence of the fact that $\sup_{g \in \Sigma^\omega} \varPhi^*(sg) \leq \sup_{g \in \Sigma^\omega} \varPhi(sg)$ for all $s\in\Sigma^*$, we have that $\inf_{s \prefix f}\sup_{g \in \Sigma^\omega} \varPhi(sg) \geq v$.
	It implies that $\varPhi^*(f) \geq v$, which contradicts the hypothesis $\varPhi^*(f) \ngeq v$.
	Hence $\varPhi^*$ is safe.
	
	Next, we prove that $\varPhi^*(f) \geq \varPhi(f)$ for all $f\in\Sigma^\omega$.
	Given $s \in \Sigma^*$, let $P_{\varPhi,s} = \{\varPhi(sg) \st g \in \Sigma^\omega\}$.
	Observe that $\varPhi^*(f) = \lim_{s \prec f} (\sup P_{\varPhi,s})$ for all $f\in\Sigma^\omega$.
	Moreover, $\varPhi(f) \in P_{\varPhi,s}$ for each $s \prec f$, and thus $\sup P_{\varPhi,s} \geq \varPhi(f)$ for each $s \prec f$, which implies $\lim_{s \prec f} (\sup P_{\varPhi,s}) \geq \varPhi(f)$, since the sequence of suprema is monotonically decreasing.
	
	Next, we prove that $\varPhi^*(f) = {\varPhi^*}^*(f)$ for all $f\in\Sigma^\omega$.
	Recall from the first paragraph that $\sup_{g \in \Sigma^\omega} \varPhi^*(sg) \leq \sup_{g \in \Sigma^\omega} \varPhi(sg)$ for all $s\in\Sigma^*$.
	So, for every $f\in\Sigma^\omega$, we have $\inf_{s \prefix f} \sup_{g \in \Sigma^\omega} \varPhi^*(sg) \leq \inf_{s \prefix f} \sup_{g \in \Sigma^\omega} \varPhi(sg)$ and thus ${\varPhi^*}^*(f) \leq \varPhi^*(f)$ for all $f\in\Sigma^\omega$.
	Since we also have ${\varPhi^*}^*(f) \geq \varPhi^*(f)$, then ${\varPhi^*}^*(f) = \varPhi^*(f)$ for all $f\in\Sigma^\omega$.
	
	Finally, we prove that $\varPhi^*$ is the least safety property that bounds $\varPhi$ from above.
	Suppose towards contradiction that there exists a safety property $\varPsi$ such that $\varPhi(f)\leq \varPsi(f)$ holds for all $f\in\Sigma^\omega$ but there exists $g \in \Sigma^\omega$ satisfying $\varPsi(g) < \varPhi^*(g)$.
	Since $\varPsi(g) \not\geq \varPhi^*(g)$ and as $\varPsi$ is safe, there exists $s \prefix g$ for which $\sup_{h\in\Sigma^\omega} \varPsi(sh) \not\geq \varPhi^*(g)$.
	Let $v = \sup_{h\in\Sigma^\omega} \varPsi(sh)$.
	Furthermore, we have $v \geq \sup_{h\in\Sigma^\omega} \varPhi(sh) $ by hypothesis.
	Consider the set $S_g = \{u \in \D \st \exists r \prefix g : \sup_{h\in\Sigma^\omega} \varPhi(rh)
	\leq u \}$ and observe that $v \in S_g$.
	By definition, $\varPhi^*(g)=\inf S_g$, implying that $v \geq \varPhi^*(g)$, which contradicts the choice of $v$.
\qed\end{proof}

\subsection{Alternative Characterizations of Quantitative Safety}
Consider a trace and its prefixes of increasing length.
For a given property, the ghost monitor from the introduction maintains, for each prefix, the $\sup$ of possible prediction values, i.e., the least upper bound of the property values for all possible infinite continuations.
The resulting sequence of monotonically decreasing suprema provides an upper bound on the eventual property value.
Moreover, for some properties, this sequence always converges to the property value.
If this is the case, then the ghost monitor can always dismiss wrong lower-bound hypotheses after finite prefixes, and vice versa.
This gives us an alternative definition for the safety of quantitative properties which, inspired by the notion of Scott continuity, was called \emph{continuity}~\cite{DBLP:conf/lics/HenzingerS21}.
We believe that \emph{upper semicontinuity} is a more appropriate term, as becomes clear when we consider the Cantor topology on $\Sigma^\omega$ and the value domain $\R \cup \{-\infty, +\infty\}$ with the standard order topology.

\begin{definition} [Upper semicontinuity~\cite{DBLP:conf/lics/HenzingerS21}]
	A property $\varPhi$ is \emph{upper semicontinuous} iff $\varPhi(f) = \lim_{s \prefix f} \sup_{g \in \Sigma^\omega} \varPhi(sg)$ for all $f \in \Sigma^\omega$. 
\end{definition}

We note that the minimal response-time property is upper semicontinuous.
\begin{example}
	Recall the minimal response-time property $\varPhi_{\min}$ from Example~\ref{ex:minresp}.
	For every trace $f\in\Sigma^\omega$, the $\varPhi_{\min}$ value is the limit of the $\pi_{\min}$ values for the prefixes of $f$.
	Therefore, $\varPhi_{\min}$ is upper semicontinuous.
\qed\end{example}

In general, a property is safe iff it maps every trace to the limit of the suprema of possible prediction values.
Moreover, we can also characterize safety properties as the properties that are equal to their safety closure.

\begin{theorem}\label{thm:safe:main}
	For every property $\varPhi$, the following statements are equivalent:
	\begin{enumerate}
		\item $\varPhi$ is safe.
		\item $\varPhi$ is upper semicontinuous.
		\item $\varPhi(f) = \varPhi^*(f)$ for all $f \in \Sigma^\omega$.
	\end{enumerate}
\end{theorem}
\begin{proof} 
		We only show the first equivalence as the other follows from the definitions.
		Assume $\varPhi$ is safe, i.e., for all $f \in \Sigma^\omega$ and $v \in \D$ if $\varPhi(f) \not \geq v$ then there exists $s \prec f$ with $\sup_{g \in \Sigma^\omega} \varPhi(sg) \not \geq v$.
		Suppose towards contradiction that $\varPhi$ is not upper semicontinuous, i.e., for some $f' \in \Sigma^\omega$ we have $\varPhi(f') < \lim_{s' \prec f'} \sup_{g \in \Sigma^\omega} \varPhi(s'g)$.
		Let $v = \lim_{s' \prec f'} \sup_{g \in \Sigma^\omega} \varPhi(s'g)$.	
		Since $\varPhi$ is safe and $\varPhi(f') \not \geq v$, there exists $r \prec f'$ such that $\sup_{g \in \Sigma^\omega} \varPhi(rg) \not \geq v$.
		Observe that for all $f \in \Sigma^\omega$ and $s_1 \prec s_2 \prec f$ we have $\sup_{g \in \Sigma^\omega} \varPhi(s_2 g) \leq \sup_{g \in \Sigma^\omega} \varPhi(s_1 g)$, i.e., the supremum is monotonically decreasing with longer prefixes.
		Therefore, we have $\lim_{s' \prec f'} \sup_{g \in \Sigma^\omega} \varPhi(s'g) \leq \sup_{g \in \Sigma^\omega} \varPhi(rg)$.
		But since $\sup_{g \in \Sigma^\omega} \varPhi(rg) \not \geq v$, we get a contradiction.
		
		Now, assume $\varPhi$ is upper semicontinuous, i.e., for all $f \in \Sigma^\omega$ we have $\varPhi(f) = \lim_{s \prec f} \sup_{g \in \Sigma^\omega} \varPhi(sg)$.
		Suppose towards contradiction that $\varPhi$ is not safe, i.e., for some $f' \in \Sigma^\omega$ and $v \in \D$ with $\varPhi(f') \not \geq v$ we have that $\sup_{g \in \Sigma^\omega} \varPhi(s'g) \geq v$ for all $s' \prec f'$.
		Since the supremum over all infinite continuations is monotonically decreasing as we observed above, we get $\lim_{s' \prec f'} \sup_{g \in \Sigma^\omega} \varPhi(s'g) \geq v$.
		However, since $\varPhi$ is upper semicontinuous, we have $\varPhi(f') = \lim_{s' \prec f'} \sup_{g \in \Sigma^\omega} \varPhi(s'g)$.
		Therefore, we obtain a contradiction to $\varPhi(f') \not \geq v$.
	\qed\end{proof}

\subsection{Related Definitions of Quantitative Safety}


In~\cite{DBLP:journals/isci/LiDL17}, the authors consider the model-checking problem for properties  on multi-valued truth domains.
They introduce the notion of multi-safety through a closure operation that coincides with our safety closure.
Formally, a property $\varPhi$ is \emph{multi-safe} iff $\varPhi(f) = \varPhi^*(f)$ for every $f \in \Sigma^\omega$.
It is easy to see the following.

\begin{proposition}
	For every property $\varPhi$, we have $\varPhi$ is multi-safe iff $\varPhi$ is safe.
\end{proposition}
	
Although the two definitions of safety are equivalent, our definition is consistent with the membership problem for quantitative automata and motivated by the monitoring of quantitative properties.

In~\cite{DBLP:conf/nfm/GorostiagaS22}, the authors extend a refinement of the safety-liveness classification for monitoring~\cite{DBLP:conf/birthday/PeledH18} to richer domains.
They introduce the notion of verdict-safety through dismissibility of values not less than or equal to the property value.
Formally, a property $\varPhi$ is \emph{verdict-safe} iff for every $f \in \Sigma^\omega$ and $v \not \leq \varPhi(f)$, there exists a prefix $s \prec f$ such that for all $g \in \Sigma^\omega$, we have $\varPhi(sg) \neq v$. 

We demonstrate that verdict-safety is weaker than safety.
Moreover, we provide a condition under which the two definitions coincide.
To achieve this, we reason about sets of possible prediction values:
for a property $\varPhi$ and $s \in \Sigma^*$, let $P_{\varPhi,s} = \{\varPhi(sf) \st f \in \Sigma^\omega\}$.

\begin{lemma}\label{lem:GS:characterization}
	A property $\varPhi$ is verdict-safe iff $\varPhi(f) = \sup (\lim_{s \prec f} P_{\varPhi,s})$ for all $f \in \Sigma^\omega$. 
\end{lemma}
\begin{proof}
	For all $f \in \Sigma^\omega$ let us define $P_f = \lim_{s \prec f} P_{\varPhi,s} = \bigcap_{s \prec f} P_{\varPhi,s}$.
	Assume $\varPhi$ is verdict-safe and suppose towards contradiction that $\varPhi(f) \neq \sup P_f$ for some $f \in \Sigma^\omega$.
	If $\varPhi(f) \not \leq \sup P_f$, then $\varPhi(f) \notin P_f$, which is a contradiction.
	Otherwise, if $\varPhi(f) < \sup P_f$, there exists $v \not \leq \varPhi(f)$ with $v \in P_f$.
	It means that there is no $s \prec f$ that dismisses the value $v \not \leq \varPhi(f)$, which contradicts the fact that $\varPhi$ is verdict-safe.
	Therefore, $\varPhi(f) = \sup P_f$ for all $f \in \Sigma^\omega$.
	
	We prove the other direction by contrapositive.
	Assume $\varPhi$ is not verdict-safe, i.e., for some $f \in \Sigma^\omega$ and $v \not \leq \varPhi(f)$, every $s \prec f$ has an extension $g \in \Sigma^\omega$ with $\varPhi(sg) = v$.
	Equivalently, for some $f \in \Sigma^\omega$ and $v \not \leq \varPhi(f)$, every $s \prec f$ satisfies $v \in P_{\varPhi,s}$.
	Then, $v \in P_f$, but since $v \not \leq \varPhi(f)$, we have $\sup P_f > \varPhi(f)$.
\qed\end{proof}

Notice that $\varPhi$ is safe iff $\varPhi(f) = \lim_{s \prec f} (\sup P_{\varPhi,s})$ for all $f \in \Sigma^\omega$.
Below we describe a property that is verdict-safe but not safe.

\begin{example}
	Let $\Sigma = \{a,b\}$.
	Define $\varPhi$ by $\varPhi(f) = 0$ if $f = a^\omega$, and $\varPhi(f) = |s|$ otherwise, where $s \prec f$ is the shortest prefix in which $b$ occurs.
	The property $\varPhi$ is verdict-safe.
	First, observe that $\D = \N \cup \{\infty\}$.
	Let $f \in \Sigma^\omega$ and $v \in \D$ with $v > \varPhi(f)$.
	If $\varPhi(f) > 0$, then $f$ contains $b$, and $\varPhi(f) = |s|$ for some $s \prec f$ in which $b$ occurs for the first time.
	After the prefix $s$, all $g \in \Sigma^\omega$ yield $\varPhi(sg) = |s|$, thus all values above $|s|$ are rejected.
	If $\varPhi(f) = 0$, then $f = a^\omega$.
	Let $v \in \D$ with $v > 0$, and consider the prefix $a^v \prec f$.
	Observe that the set of possible prediction values after reading $a^v$ is $\{0, v+1, v+2, \ldots\}$, therefore $a^v$ allows the ghost monitor to reject the value $v$.
	However, $\varPhi$ is not safe because, although $\varPhi(a^\omega) = 0$, for every $s \prec a^\omega$, we have $\sup_{g \in \Sigma^\omega} \varPhi(sg) = \infty$.
\qed\end{example}

The separation is due to the fact that for some finite traces, the $\sup$ of possible prediction values cannot be realized by any future.
Below, we present a condition that prevents such cases.
\begin{definition}[Supremum closedness]	
	A property $\varPhi$ is \emph{$\sup$-closed} iff for every $s \in \Sigma^*$ we have $\sup P_{\varPhi, s} \in P_{\varPhi, s}$.
\end{definition}

We remark that the minimal response-time property is $\sup$-closed.
\begin{example}
	The safety property minimal response-time $\varPhi_{\min}$ from Example~\ref{ex:minresp} is $\sup$-closed.
	This is because, for every $s \in \Sigma^*$, the continuation $\gra^\omega$ realizes the value $\sup_{g \in \Sigma^\omega} \varPhi(sg)$.
\qed\end{example}

Recall from the introduction the ghost monitor that maintains the $\sup$ of possible prediction values.
For monitoring $\sup$-closed properties this suffices; otherwise the ghost monitor also needs to maintain whether or not the supremum of the possible prediction values is realizable by some future continuation.
In general, we have the following for every $\sup$-closed property.

\begin{lemma}\label{lemma:value:closed}
  For every $\sup$-closed property $\varPhi$ and for all $f \in \Sigma^\omega$, we have $\lim_{s \prec f} (\sup P_{\varPhi, s}) = \sup (\lim_{s \prec f} P_{\varPhi,s})$.
\end{lemma}
\begin{proof}
	Note that $\lim_{s \prec f} (\sup P_{\varPhi, s}) \geq \sup (\lim_{s \prec f} P_{\varPhi,s})$ holds in general, and we want to show that $\lim_{s \prec f} (\sup P_{\varPhi, s}) \leq \sup (\lim_{s \prec f} P_{\varPhi,s})$ holds for every value-closed $\varPhi$.
	Let $f \in \Sigma^\omega$.
	Since the sequence $(P_{\varPhi,s})_{s \prec f}$ of sets is monotonically decreasing and $P_{\varPhi,s}$ is closed for every $s \in \Sigma^*$, we have $\sup P_{\varPhi,r} \in P_{\varPhi,s}$ for every $s,r \in \Sigma^*$ with $s \preceq r$. 
	Moreover, $\lim_{s \prec f} (\sup P_{\varPhi,s}) \in P_{\varPhi,r}$ for every $r \in \Sigma^*$ with $r \prec f$.
	Then, by definition, we have $\lim_{s \prec f} (\sup P_{\varPhi,s}) \in \lim_{s \prec f} P_{\varPhi,s}$, and therefore $\lim_{s \prec f} (\sup P_{\varPhi,s}) \leq \sup (\lim_{s \prec f} P_{\varPhi,s})$.
\qed\end{proof}

As a consequence of the lemmas above, we get the following.
	
\begin{theorem}
	A $\sup$-closed property $\varPhi$ is safe iff $\varPhi$ is verdict-safe.
\end{theorem}

\section{The Quantitative Safety-Progress Hierarchy}
Our quantitative extension of safety closure allows us to build a Borel hierarchy, which is a quantitative extension of the boolean safety-progress hierarchy~\cite{ChangMP93}.
First, we show that safety properties are closed under pairwise $\min$ and $\max$.

\begin{proposition}\label{prop:safe:clorure}
	For every value domain $\D$, the set of safety properties over $\D$ is closed under $\min$ and $\max$.
\end{proposition}
\begin{proof}
	First, we prove the closure under $\min$.
	Consider the two safety properties $\varPhi_1$, $\varPhi_2$ and let $\varPhi$ be their pairwise minimum, i.e., $\varPhi(f) = \min(\varPhi_1(f),\varPhi_2(f))$ for all $f\in \Sigma^\omega$.
	Suppose towards contradiction that $\varPhi$ is not safe, i.e., for some $f \in \Sigma^\omega$ and $v \in \D$ such that $\varPhi(f) \not \geq v$ and $\sup_{g \in \Sigma^\omega} \varPhi(sg) \geq v$ for all $s \prec f$.
	Observe that $\varPhi(f) \not \geq v$  implies $\varPhi_1(f) \not \geq v$ or $\varPhi_2(f) \not \geq v$.
	We assume without loss of generality that $\varPhi_1(f) \not \geq v$ holds.
	Thanks to the safety of $\varPhi_1$, there exists $r \prec f$ such that $\sup_{g \in \Sigma^\omega} \varPhi_1(rg) \not \geq v$.
	Since $\varPhi_1(rg) \geq \varPhi(rg)$ for all $g \in \Sigma^\omega$, we have that $\sup_{g \in \Sigma^\omega} \varPhi_1(rg) \geq \sup_{g \in \Sigma^\omega}\varPhi(rg) \geq v$.
	This implies that $\sup_{g \in \Sigma^\omega} \varPhi_1(rg) \geq v$, which yields a contradiction.
	
	Then, we prove the closure under $\max$.
	Consider the two safety properties $\varPsi_1$, $\varPsi_2$ and let $\varPsi$ be their pairwise maximum, i.e., $\varPsi(f) = \max(\varPsi_1(f),\varPsi_2(f))$ for all $f\in \Sigma^\omega$.
	Suppose towards contradiction that $\varPsi$ is not safe, i.e., for some $f \in \Sigma^\omega$ and $v \in \D$, we have $\varPsi(f) \not \geq v$ and $\sup_{g \in \Sigma^\omega} \varPsi(sg) \geq v$ for all $s \prefix f$.
	Due to the safety of both $\varPsi_1$ and $\varPsi_2$, we get that for all $f \in \Sigma^\omega$ and $v \in \D$ if $\varPsi_i(f) \not \geq v_i$, there is $s_i \prec f$ such that $\sup_{g \in \Sigma^\omega} \varPsi(s_i g) \not \geq v_i$ with $i\in\{1,2\}$.
	Combining the two statements, we get for all $f \in \Sigma^\omega$ and $v\in \D$ if $\max(\varPsi_1(f), \varPsi_2(f))\not \geq \max(v, v)$, then there exists $s \prec f$ such that $\max(\sup_{g \in \Sigma^\omega} \varPsi_1(sg), \sup_{g \in \Sigma^\omega} \varPsi_2(sg)) \not \geq \max(v,v)$.
	In particular, we have that $\max(\sup_{g \in \Sigma^\omega} \varPsi_1(sg), \sup_{g \in \Sigma^\omega} \varPsi_2(sg)) \not \geq v$ holds since $\max(\varPsi_1(f), \varPsi_2(f)) = \varPsi(f) \not \geq v = \max(v, v)$.
	It is well know that $\sup (X \cup Y) = \max (\sup X, \sup Y)$ for all $X, Y \subseteq \D$, implying $\sup_{g \in \Sigma^\omega} (\max(\varPsi_1(sg), \varPsi_2(sg))) = \max(\sup_{g \in \Sigma^\omega} \varPsi_1(sg), \sup_{g \in \Sigma^\omega} \varPsi_2(sg))$.
	Consequently $\sup_{g \in \Sigma^\omega} \max(\varPsi_1(sg), \varPsi_2(sg))$ $=$ $\sup_{g \in \Sigma^\omega} \varPsi(sg) \not \geq v$, which yields a contradiction.	
\qed\end{proof}

The boolean safety-progress classification of properties is a Borel hierarchy built from the Cantor topology of traces.
Safety and co-safety properties lie on the first level, respectively corresponding to the closed sets and open sets of the topology.
The second level is obtained through countable unions and intersections of properties from the first level:
persistence properties are countable unions of closed sets, while response properties are countable intersections of open sets.
We generalize this construction to the quantitative setting.

In the boolean case, each property class is defined through an operation that takes a set $S \subseteq \Sigma^*$ of finite traces and produces a set $P \subseteq \Sigma^\omega$ of infinite traces.
For example, to obtain a co-safety property from $S \subseteq \Sigma^*$, the corresponding operation yields $S\Sigma^\omega$.
Similarly, we formalize each property class by a value function.
For this, we define the notion of \emph{limit property}.
\begin{definition}[Limit property]
	A property $\varPhi : \Sigma^\omega \to \D$ is a \emph{limit property} iff there exists a finitary property $\pi : \Sigma^* \rightarrow \D$ and a value function $\ell : \D^\omega \to \D$ such that $\varPhi(f) = \ell_{s \prec f} \pi(s)$ for all $f \in \Sigma^\omega$.
	We denote this by $\varPhi = (\pi,\ell)$, and write $\varPhi(s)$ instead of $\pi(s)$.
	In particular, if $\varPhi = (\pi,\ell)$, where $\ell \in \{ \inf, \sup, \liminf, \limsup\}$, then $\varPhi$ is an \emph{$\ell$-property}.
\end{definition}

To account for the value functions that construct the first two levels of the safety-progress hierarchy, we start our investigation with $\inf$- and $\sup$-properties and later focus on $\liminf$- and $\limsup$- properties~\cite{DBLP:journals/tocl/ChatterjeeDH10}.

\subsection{Infimum and Supremum Properties}
Let us start with an example by demonstrating that the minimal response-time property is an $\inf$-property.
\begin{example} \label{ex:infresp}
	Recall the safety property $\varPhi_{\min}$ of minimal response time from Example~\ref{ex:minresp}.
	We can equivalently define $\varPhi_{\min}$ as a limit property by taking the finitary property $\pi_{\text{last}}$ and the value function $\inf$.
	As discussed in Example~\ref{ex:minresp}, the function $\pi_{\text{last}}$ outputs the response time for the last request when all requests are granted, and $\infty$ when there is a pending request or no request.
	Then $\inf_{s \prec f} \pi_{\text{last}}(s) = \varPhi_{\min}(f)$ for all $f \in \Sigma^\omega$, and therefore $\varPhi_{\min} = (\pi_{\text{last}}, \inf)$.
\qed\end{example}

In fact, the safety properties coincide with $\inf$-properties.
\begin{theorem}\label{thm:safe:inf}
	A property $\varPhi$ is safe iff $\varPhi$ is an $\inf$-property.
\end{theorem}
\begin{proof}
	Assume $\varPhi$ is safe.
	By Theorem~\ref{thm:safe:main}, we have $\varPhi(f) = \inf_{s \prec f} \sup_{g \in \Sigma^\omega} \varPhi(sg)$ for all $f \in \Sigma^\omega$.
	Then, simply taking $\pi(s) = \sup_{g \in \Sigma^\omega} \varPhi(sg)$ for all $s \in \Sigma^*$ yields that $\varPhi$ is an $\inf$ property.
	
	Now, assume $\varPhi$ is an $\inf$ property, and suppose towards contradiction that $\varPhi$ is not safe.
	In other words, let $\varPhi = (\pi, \inf)$ for some finitary property $\pi : \Sigma^* \to \D$ and suppose $\inf_{s \prec f'} \sup_{g \in \Sigma^\omega} \varPhi(sg) > \varPhi(f') = \inf_{s \prec f'} \pi(s)$ for some $f' \in \Sigma^\omega$.
	Let $s \in \Sigma^*$ and note that $\sup_{g \in \Sigma^\omega} \varPhi(sg) = \sup_{g \in \Sigma^\omega} (\inf_{r \prec sg} \pi(r))$ by definition.
	Moreover, for every $g \in \Sigma^\omega$, notice that $\inf_{r \prec sg} \pi(r) \leq \pi(s)$ since $s \prec sg$.
	Then, we obtain $\sup_{g \in \Sigma^\omega} \varPhi(sg) \leq \pi(s)$ for every $s \in \Sigma^*$.
	In particular, this is also true for all $s \prec f'$.
	Therefore, we get $\inf_{s \prec f'} \sup_{g \in \Sigma^\omega} \varPhi(sg) \leq \inf_{s \prec f'} \pi(s)$, which contradicts to our initial supposition.
\qed\end{proof}

Defining the minimal response-time property as a limit property, we observe the following relation between its behavior on finite traces and infinite traces.
\begin{example}
	Consider the property $\varPhi_{\min} = (\pi_{\text{last}}, \inf)$ from Example~\ref{ex:infresp}.
	Let $f \in \Sigma^\omega$ and $v \in \D$.
	Observe that if the minimal response time of $f$ is at least $v$, then the last response time for each prefix $s \prec f$ is also at least $v$.
	Conversely, if the minimal response time of $f$ is below $v$, then there is a prefix $s \prec f$ for which the last response time is also below $v$.
\qed\end{example}

In light of this observation, we provide another characterization of safety properties, explicitly relating the specified behavior of the limit property on finite and infinite traces.
\begin{theorem}\label{thm:safe:inflooking}
	A property $\varPhi:\Sigma^\omega\rightarrow\D$ is safe iff $\varPhi$ is a limit property such that for every $f \in \Sigma^\omega$ and value $v \in \D$, we have $\varPhi(f) \geq v$ iff $\varPhi(s) \geq v$ for all $s \prec f$.
\end{theorem}
\begin{proof}
	Assume $\varPhi$ is safe.
	Then we know by Theorem~\ref{thm:safe:inf} that $\varPhi$ is an $\inf$ property, i.e., $\varPhi = (\pi, \inf)$ for some finitary property $\pi : \Sigma^* \to \D$, and thus a limit property.
	Suppose towards contradiction that for some $f \in \Sigma^\omega$ and $v \in \D$ we have (i) $\varPhi(f) \geq v$ and $\pi(s) \not \geq v$ for some $s \prec f$, or (ii) $\varPhi(f) \not \geq v$ and $\pi(s) \geq v$ for every $s \prec f$.
	One can easily verify that (i) yields a contradiction, since if for some $s \prec f$ we have $\pi(s) \not \geq v$ then $\inf_{s \prec f} \pi(s) = \varPhi(f) \not \geq v$.
	Similarly, (ii) also yields a contradiction, since if $\varPhi(f) = \inf_{s \prec f} \pi(s) \not \geq v$ then there exists $s \prec f$ such that $\pi(s) \not \geq v$.
	
	Now, assume $\varPhi = (\pi,\ell)$ for some finitary property $\pi$ and value function $\ell$ such that for every $f \in \Sigma^\omega$ and value $v \in \D$ we have $\varPhi(f) \geq v$ iff $\pi(s) \geq v$ for every $s \prec f$.
	We claim that $\varPhi(f) = \inf_{s \prec f} \pi(s)$ for every $f \in \Sigma^\omega$.
	Suppose towards contradiction that the equality does not hold for some trace.
	If $\varPhi(f) \not \geq \inf_{s \prec f} \pi(s)$ for some $f \in \Sigma^\omega$, let $v = \inf_{s \prec f} \pi(s)$ and observe that (i) $\varPhi(f) \not \geq v$, and (ii) $\inf_{s \prec f} \pi(s) \geq v$.
	However, while (i) implies $\pi(s) \not \geq v$ for some $s \prec f$ by hypothesis, (ii) implies $\pi(s) \geq v$ for all $s \prec f$, resulting in a contradiction.
	The case where $\varPhi(f) \not \leq \inf_{s \prec f} \pi(s)$ for some $f \in \Sigma^\omega$ is similar.
	It means that $\varPhi$ is an $\inf$ property.
	Therefore, $\varPhi$ is safe by Theorem~\ref{thm:safe:inf}.
\qed\end{proof}

Recall that a safety property allows rejecting wrong lower-bound hypotheses with a finite witness, by assigning a tight upper bound to each trace.
We define co-safety properties symmetrically: a property $\varPhi$ is co-safe iff every wrong hypothesis of the form $\varPhi(f) \leq v$ has a finite witness $s \prec f$.

\begin{definition}[Co-safety]
	A property $\varPhi : \Sigma^\omega \rightarrow \D$ is \emph{co-safe} iff for every $f\in\Sigma^\omega$ and value $v\in\D$ with $\varPhi(f) \not\leq v$, there exists a prefix $s \prefix f$ such that $\inf_{g \in \Sigma^\omega} \varPhi(sg) \not\leq v$.
\end{definition}

We note that our definition generalizes boolean co-safety, and thus a dual of Proposition~\ref{prop:safetyboolean} holds also for co-safety.
Moreover, we analogously define the notions of co-safety closure and lower semicontinuity.

\begin{definition}[Co-safety closure]
  The \emph{co-safety closure} of a property $\varPhi$ is the property $\varPhi_*(f)$ defined by $\varPhi_*(f) = \sup_{s \prec f} \inf_{g \in \Sigma^\omega} \varPhi(sg)$ for all $f\in\Sigma^\omega$. 
\end{definition}

\begin{definition} [Lower semicontinuity~\cite{DBLP:conf/lics/HenzingerS21}]
  A property $\varPhi$ is \emph{lower semicontinuous} iff $\varPhi(f) = \lim_{s \prefix f} \inf_{g \in \Sigma^\omega} \varPhi(sg)$ for all $f \in \Sigma^\omega$. 
\end{definition}

Now, we define and investigate the \emph{maximal response-time} property.
In particular, we show that it is a $\sup$-property that is co-safe and lower semicontinuous.

\begin{example} \label{ex:maxresp}
	Let $\Sigma = \{\req, \gra, \tick, \other\}$ and $\D = \N \cup \{\infty\}$.
	We define the maximal response-time property $\varPhi_{\max}$ through a finitary property that computes the current response time for each finite trace and the value function $\sup$.
	In particular, for all $s \in \Sigma^*$, let $\pi_{\text{curr}}(s) = |s|_\tick - |r|_\tick$, where $r \preceq s$ is the longest prefix of $s$ without pending $\req$; then $\varPhi_{\max} = (\pi_{\text{curr}},\sup)$.
	Note the contrast between $\pi_{\text{curr}}$ and $\pi_{\text{last}}$ from Example~\ref{ex:minresp}.
	While $\pi_{\text{curr}}$ takes an optimistic view of the future and assumes the $\gra$ will follow immediately, $\pi_{\text{last}}$ takes a pessimistic view and assumes the $\gra$ will never follow.
	Let $f \in \Sigma^\omega$ and $v \in \D$. 
	If the maximal response time of $f$ is greater than $v$, then for some prefix $s \prec f$ the current response time is greater than $v$ also, which means that, no matter what happens in the future, the maximal response time is greater than $v$ after observing $s$.
	Therefore, $\varPhi_{\max}$ is co-safe.
	By a similar reasoning, the sequence of greatest lower bounds of possible prediction values over the prefixes converges to the property value.
	In other words, we have $\lim_{s \prec f} \inf_{g \in \Sigma^\omega} \varPhi_{\max}(sg) = \varPhi_{\max}(f)$ for all $f \in \Sigma^\omega$.
	Thus $\varPhi_{\max}$ is also lower semicontinuous, and it equals its co-safety closure.
	Now, consider the complementary property $\overline{\varPhi_{\max}}$, which maps every trace to the same value as $\varPhi_{\max}$ on a domain where the order is reversed.
	It is easy to see that $\overline{\varPhi_{\max}}$ is safe.
	Finally, recall the ghost monitor from the introduction, which maintains the infimum of possible prediction values for the maximal response-time property.
	Since the maximal response-time property is $\inf$-closed, the output of the ghost monitor after every prefix is realizable by some future continuation, and that output is $\pi_{\max}(s) = \max_{r \preceq s} \pi_{\text{curr}}(r)$ for all $s \in \Sigma^*$.
\qed\end{example}

Generalizing the observations in the example above, we obtain the following characterizations due to the duality between safety and co-safety.

\begin{theorem} \label{thm:cosafe}
	For every property $\varPhi : \Sigma^\omega \to \D$, the following are equivalent.
	\begin{enumerate}
		\item $\varPhi$ is co-safe.
		\item $\varPhi$ is lower semicontinuous.
		\item $\varPhi(f) = \varPhi_*(f)$ for every $f \in \Sigma^\omega$.
		\item $\varPhi$ is a $\sup$-property.
		\item $\varPhi$ is a limit property such that for every $f \in \Sigma^\omega$ and value $v \in \D$, we have $\varPhi(f) \leq v$ iff $\varPhi(s) \leq v$ for all $s \prec f$.
		\item $\overline{\varPhi}$ is safe.
	\end{enumerate}
\end{theorem}
\vspace*{-0.5cm}
\subsection{Limit Inferior and Limit Superior Properties}
Let us start with an observation on the minimal response-time property.
\begin{example}
	Recall once again the minimal response-time property $\varPhi_{\min}$ from Example~\ref{ex:minresp}.
	In the previous subsection, we presented an alternative definition of $\varPhi_{\min}$ to establish that it is an $\inf$-property.
	Observe that there is yet another equivalent definition of $\varPhi_{\min}$	which takes the monotonically decreasing finitary property $\pi_{\min}$ from Example~\ref{ex:minresp} and pairs it with either the value function $\liminf$, or with $\limsup$.
	Hence $\varPhi_{\min}$ is both a $\liminf$- and a $\limsup$-property.
\qed\end{example}

Before moving on to investigating $\liminf$- and $\limsup$-properties more closely, we show that the above observation can be generalized.
\begin{theorem}\label{thm:inf:liminf:limsup}
  Every $\ell$-property $\varPhi$, for $\ell\in\{\inf, \sup\}$, is both a $\liminf$- and a $\limsup$-property.
\end{theorem}
\begin{proof}
	Let $\varPhi = (\pi, \inf)$ and define an alternative finitary property as follows:
	$\pi'(s) = \min_{r \preceq s} \pi(s)$.
	One can confirm that $\pi'$ is monotonically decreasing and thus $\lim_{s \prec f} \pi'(s) = \inf_{s \prec f} \pi(s)$ for every $f \in \Sigma^\omega$.
	Then, letting $\varPhi_1 = (\pi',\liminf)$ and $\varPhi_2 = (\pi',\limsup)$, we obtain that $\varPhi(f) = \varPhi_1(f) = \varPhi_2(f)$ for all $f \in \Sigma^\omega$.
	For $\ell = \sup$ we use $\max$ instead of $\min$.
\qed\end{proof}

An interesting response-time property beyond safety and co-safety arises when we remove extreme values: instead of minimal response time, consider the property that maps every trace to a value that bounds from below, not all response times, but all of them from a point onward (i.e., all but finitely many).
We call this property  \emph{tail-minimal response time}.
\begin{example} \label{ex:liminfresp}
	Let $\Sigma = \{\req,\gra,\tick,\other\}$ and $\pi_{\text{last}}$ be the finitary property from Example~\ref{ex:minresp} that computes the last response time.
	We define the tail-minimal response-time property as $\varPhi_{\text{tmin}} = (\pi_{\text{last}},\liminf)$.
	Intuitively, it maps each trace to the least response time over all but finitely many requests.
	This property is interesting as a performance measure, because it focuses on the long-term performance by ignoring finitely many outliers.	
	Consider $f \in \Sigma^\omega$ and $v \in \D$.
	Observe that if the tail-minimal response time of $f$ is at least $v$, then there is a prefix $s \prec f$ such that for all longer prefixes $s \preceq r \prec f$, the last response time in $r$ is at least $v$, and vice versa.
\qed\end{example}

Similarly as for $\inf$-properties, we characterize $\liminf$-properties through a relation between property behaviors on finite and infinite traces.

\begin{theorem}\label{thm:liminf:liminflooking}
	A property $\varPhi:\Sigma^\omega\rightarrow\D$ is a $\liminf$-property iff $\varPhi$ is a limit property such that for every $f \in \Sigma^\omega$ and value $v \in \D$, we have $\varPhi(f) \geq v$ iff there exists $s \prec f$ such that for all $s \preceq r \prec f$, we have $\varPhi(r) \geq v$.
\end{theorem}
\begin{proof}
	Assume $\varPhi$ is a $\liminf$ property, i.e., $\varPhi = (\pi,\liminf)$ for some finitary property $\pi : \Sigma^* \to \D$.
	Suppose towards contradiction that for some $f \in \Sigma^\omega$ and $v \in \D$ we have
	(i) $\varPhi(f) \geq v$ and for all $s \prec f$ there exists $s \preceq r \prec f$ such that $\pi(r) \not \geq v$, or
	(ii) $\varPhi(f) \not \geq v$ and there exists $s \prec f$ such that for all $s \preceq r \prec f$ we have $\pi(r) \geq v$.
	One can easily verify that (i) yields a contradiction, since if for all $s \prec f$ there exists $s \preceq r \prec f$ with $\varPhi(r) \not \geq v$, then $\liminf_{s \prec f} \pi(s) = \varPhi(f) \not \geq v$.
	Similarly, (ii) also yields a contradiction, since if there exists $s \prec f$ such that for all $s \preceq r \prec f$ we have $\pi(r) \geq v$ then $\liminf_{s \prec f} \pi(s) = \varPhi(f) \geq v$.
	
	Now, assume $\varPhi = (\pi,\ell)$ for some finitary property $\pi$ and value function $\ell$ such that for every $f \in \Sigma^\omega$ and value $v \in \D$ we have $\varPhi(f) \geq v$ iff there exists $s \prec f$ such that for all $s \preceq r \prec f$ we have $\pi(r) \geq v$.
	We claim that $\varPhi(f) = \liminf_{s \prec f} \pi(s)$ for every $f \in \Sigma^\omega$.
	Suppose towards contradiction that the equality does not hold for some trace.
	If $\varPhi(f) \not \geq \liminf_{s \prec f} \pi(s)$ for some $f \in \Sigma^\omega$, let $v = \liminf_{s \prec f} \pi(s)$ and observe that (i) $\varPhi(f) \not \geq v$, and (ii) $\liminf_{s \prec f} \pi(s) \geq v$.
	However, by hypothesis, (i) implies that for all $s \prec f$ there exists $s \preceq r \prec f$ with $\pi(r) \not \geq v$, which means that $\liminf_{s \prec f} \pi(s) \not \geq v$, resulting in a contradiction to (ii).
	The case where $\varPhi(f) \not \leq \liminf_{s \prec f} \pi(s)$ for some $f \in \Sigma^\omega$ is similar.
	Therefore, $\varPhi$ is a $\liminf$ property.
\qed\end{proof}

Next, we show that $\liminf$-properties are closed under pairwise minimum.

\begin{proposition}\label{prop:liminf:clorure}
	For every value domain $\D$, the set of $\liminf$-properties over $\D$ is closed under $\min$.
\end{proposition}
\begin{proof} 	Consider two $\liminf$-properties $\varPhi_1 = (\pi_1, \liminf)$, $\varPhi_2 = (\pi_2, \liminf)$ and let $\varPhi$ be as follows: $\varPhi = (\pi, \liminf)$ where $\pi(s) = \min (\pi_1(s), \pi_2(s))$ for all $s\in \Sigma^*$.
	We now prove that $\varPhi(f) = \min(\varPhi_1(f), \varPhi_2(f))$.

	Suppose towards contradiction that $\min(\varPhi_1(f), \varPhi_2(f)) \ngeq \varPhi(f)$ for some $f\in\Sigma^\omega$.
	Observe that for all $g\in\Sigma^\omega$ and $v\in\D$, if $\min(\varPhi_1(g), \varPhi_2(g)) \ngeq v$ then $\varPhi_1(g) \ngeq v$ or $\varPhi_2(g) \ngeq v$.
	We assume without loss of generality that $\varPhi_1(f) \ngeq \varPhi(f)$.
	By Theorem~\ref{thm:liminf:liminflooking}, $\varPhi_1(f) \ngeq \varPhi(f)$ implies that for all $\hat{s} \prec f$, there exists $\hat{s} \prefixeq \hat{r} \prefix f$ such that $\varPhi_1(\hat{r}) \ngeq \varPhi(f)$.
	Dually, $\varPhi(f) \geq \varPhi(f)$ implies that there exists $\tilde{s} \prec f$ such that $\varPhi(r) \geq \varPhi(f)$ for all $\tilde{s} \prefixeq r \prefix f$.
	In particular, there exists $\tilde{s} \prefixeq \tilde{r} \prefix f$ such that $\varPhi_1(\hat{r}) \ngeq \varPhi(f)$ and $\varPhi(\hat{r}) \geq \varPhi(f)$.
	By the definition of $\min$, we have that $\varPhi_1(\hat{r}) \geq \varPhi(\hat{r}) \geq \varPhi(f)$ which contradicts that $\varPhi_1(\hat{r}) \ngeq \varPhi(f)$.
	Hence, we proved that $\min(\varPhi_1(f), \varPhi_2(f)) \geq \varPhi(f)$ for all $f\in\Sigma^\omega$.
	
	Suppose towards contradiction that $\varPhi(f) \ngeq \min(\varPhi_1(f), \varPhi_2(f))$ for some $f\in\Sigma^\omega$.
	In particular, $\liminf_{s \prec f} \min(\pi_1(s), \pi_2(s)) \ngeq \min(\varPhi_1(f), \varPhi_2(f))$.
	Observe that for all $s\in\Sigma^*$ and $v\in\D$, if $\min(\pi_1(s), \pi_2(s)) \ngeq v$ then $\pi_1(s) \ngeq v$ or $\pi_2(s) \ngeq v$.
	We assume without loss of generality that $|\{s \st \exists s \prefixeq r \prefix f, \pi_1(r) \ngeq \min(\varPhi_1(f), \varPhi_2(f))\}| = \infty$, or equivalently for all $s \prec f$, there exists $s \prefixeq r \prefix f$ such that $\pi_1(r) \ngeq \min(\varPhi_1(f), \varPhi_2(f))$.
	By Theorem~\ref{thm:liminf:liminflooking}, we get $\varPhi_1(f) \ngeq \min(\varPhi_1(f), \varPhi_2(f))$.
	By the definition of $\min$, we have that $\varPhi_1(f) \geq \min(\varPhi_1(f), \varPhi_2(f))$ which contradicts that $\varPhi_1(f) \ngeq \min(\varPhi_1(f), \varPhi_2(f))$.
	Hence, we proved that $\varPhi(f) \geq \min(\varPhi_1(f), \varPhi_2(f))$ for all $f\in\Sigma^\omega$.
\qed\end{proof}

Now, we show that the tail-minimal response-time property can be expressed as a countable supremum of $\inf$-properties.

\begin{example}
	Let $i \in \N$ and define $\pi_{i,\text{last}}$ as a finitary property that imitates $\pi_{\text{last}}$ from Example~\ref{ex:minresp}, but ignores the first $i$ observations of every finite trace.
	Formally, for $s \in \Sigma^*$, we define $\pi_{i,\text{last}}(s) = \pi_{\text{last}}(r)$ for $s = s_i r$ where $s_i \preceq s$ with $|s_i| = i$, and $r \in \Sigma^*$.
	Observe that an equivalent way to define $\varPhi_{\text{tmin}}$ from Example~\ref{ex:liminfresp} is $\sup_{i \in \N} (\inf_{s \prec f} (\pi_{i,\text{last}}(s)))$ for all $f \in \Sigma^\omega$.
	Intuitively, for each $i \in \N$, we obtain an $\inf$-property that computes the minimal response time of the suffixes of a given trace.
	Taking the supremum over these, we obtain the greatest lower bound on all but finitely many response times.
\qed\end{example}

We generalize this observation and show that every $\liminf$-property is a countable supremum of $\inf$-properties.

\begin{theorem}\label{thm:liminf:supinf}
	Every $\liminf$-property is a countable supremum of $\inf$-properties.
\end{theorem}
\begin{proof}
	Let $\varPhi = (\pi,\liminf)$.
	For each $i \in \N$ let us define $\varPhi_i = (\pi_i,\inf)$ where $\pi_i$ is as follows:
	$\pi_i(s) = \top$ if $|s| < i$, and $\pi_i(s) = \pi(s)$ otherwise.
	We claim that $\varPhi(f) = \sup_{i \in \N} \varPhi_i(f)$ for all $f \in \Sigma^\omega$.
	Expanding the definitions, observe that the claim is $\liminf_{s \prec f} \pi(s) = \sup_{i \in \N} \inf_{s \prec f} \pi_i(s)$.
	Due to the definition of $\liminf$, the left-hand side is equal to $\sup_{i \in \N} \inf_{s \prec f \land |s| \geq i} \pi(s)$.
	Moreover, due to the definition of $\pi_i$, this is equal to the right-hand side.
\qed\end{proof}

We would also like to have the converse of Theorem~\ref{thm:liminf:supinf}, i.e., that every countable supremum of $\inf$-properties is a $\liminf$-property.
Currently, we are able to show only the following.

\begin{theorem}\label{thm:liminf:leq}
  For every infinite sequence $(\varPhi_i)_{i \in \N}$ of $\inf$-properties, there is a $\liminf$-property $\varPhi$ such that $\sup_{i \in \N} \varPhi_i(f) \leq \varPhi(f)$.
\end{theorem}
\begin{proof}
	For each $i \in \N$, let $\varPhi_i = (\pi_i, \inf)$ for some finitary property $\pi_i$.
	We assume without loss of generality that each $\pi_i$ is monotonically decreasing.
	Let $\varPhi = (\pi, \liminf)$ where $\pi(s) = \max_{i \leq |s|} \pi_i(s)$ for all $s \in \Sigma^*$.
	We want to show that $\sup_{i \in \N} \varPhi_i(f) \leq \varPhi(f)$ for all $f \in \Sigma^\omega$.
	Expanding the definitions, observe that the claim is the following: $\sup_{i \in \N}( \inf_{s \prec f} \pi_i(s)) \leq \liminf_{s \prec f} (\max_{i \leq |s|} \pi_i(s))$ for all $f \in \Sigma^\omega$.
	
	Let $f \in \Sigma^\omega$, and for each $k \in \N$, let $x_k = \max_{i \leq k} \inf_{s \prec f} \pi_i(s)$ and $y_k = \max_{i \leq k} \pi_i(s_k)$ where $s_k \prec f$ with $|s_k| = k$.
	Observe that we have $x_k \leq y_k$ for all $k \in \N$.
	Then, we have $\liminf_{k \to \infty} x_k \leq \liminf_{k \to \infty} y_k$.
	Moreover, since the sequence $(x_k)_{k \in \N}$ is monotonically decreasing, we can replace the $\liminf$ on the left-hand side with $\lim$ to obtain the following:
	$\lim_{k \to \infty} \max_{i \leq k} \inf_{s \prec f} \pi_i(s) \leq \liminf_{k \to \infty} \max_{i \leq k} \pi_i(s_k)$. 
	Then, rewriting the expression concludes the proof by giving us  $\sup_{i \in \N} (\inf_{s \prec f} \pi_i(s)) \leq \liminf_{s \prec f} (\max_{i \leq |s|} \pi_i(s))$. 
\qed\end{proof}

We conjecture that some $\liminf$-property that satisfies Theorem~\ref{thm:liminf:leq} is also a lower bound on the countable supremum that occurs in the theorem.
This, together with Theorem~\ref{thm:liminf:leq}, would imply the converse of Theorem~\ref{thm:liminf:supinf}.
Proving the converse of Theorem~\ref{thm:liminf:supinf} would give us, thanks to the following duality, that the $\liminf$- and $\limsup$-properties characterize the second level of the Borel hierarchy of the topology induced
by the safety closure operator.

\begin{proposition}
	A property $\varPhi$ is a $\liminf$-property iff its complement $\overline{\varPhi}$ is a $\limsup$-property.
\end{proposition}

\section{Quantitative Liveness}
Similarly as for safety, we take the perspective of the quantitative membership problem to define liveness:
a property $\varPhi$ is live iff, whenever a property value is less than $\top$, there exists a value $v$ for which the wrong hypothesis $\varPhi(f) \geq v$ can never be dismissed by any finite witness $s \prec f$.

\begin{definition}[Liveness]
	A property $\varPhi : \Sigma^\omega \to \D$ is \emph{live} iff for all $f \in \Sigma^\omega$, if $\varPhi(f) < \top$, then there exists a value $v \in \D$ such that $\varPhi(f) \not \geq v$ and for all prefixes $s \prec f$, we have $\sup_{g \in \Sigma^\omega} \varPhi(sg) \geq v$.
\end{definition}

An equivalent definition can be given through the safety closure.

\begin{theorem} \label{thm:livenessclosure}
	A property $\varPhi$ is live iff $\varPhi^*(f) > \varPhi(f)$ for every $f \in \Sigma^\omega$ with $\varPhi(f) < \top$.
\end{theorem}
\begin{proof}
	First, suppose $\varPhi$ is live.
	Let $v$ be as in the definition of liveness, and observe that, by definition, we have $\varPhi^*(f) \geq v$ for all $f \in \Sigma^\omega$.
	Moreover, since $v \not \leq \varPhi(f)$, we are done.
	Now, suppose $\varPhi^*(f) > \varPhi(f)$ for every $f \in \Sigma^\omega$ with $\varPhi(f) < \top$.
	Let $f \in \Sigma^\omega$ be such a trace, and let $v = \varPhi^*(f)$.
	It is easy to see that $v$ satisfies the liveness condition since $\varPhi^*(f) = \inf_{s \prec f} \sup_{g \in \Sigma^\omega} \varPhi(sg) > \varPhi(f)$.
\qed\end{proof}

We show that liveness properties are closed under pairwise $\max$.
This is not the case for pairwise $\min$ as we will later present a liveness-liveness decomposition for every quantitative property (Theorem~\ref{thm:decomp}).

\begin{proposition}\label{prop:live:closure}
	For every value domain $\D$, the set of liveness properties over $\D$ is closed under $\max$.
\end{proposition}
\begin{proof}
	Consider two liveness properties $\varPhi_1$, $\varPhi_2$ and let $\varPhi$ be their pairwise maximum, i.e., $\varPhi(f) = \max(\varPhi_1(f),\varPhi_2(f))$ for all $f\in \Sigma^\omega$.
	We show that $\varPhi$ fulfills the liveness definition for any given $f \in \Sigma^\omega$.
	If $\varPhi_1(f) = \top$ or $\varPhi_2(f) = \top$ then $\varPhi(f) = \top$.
	Otherwise, for each $i \in \{1, 2\}$, there exists $v_i$ such that $\varPhi_i(f) \not\geq v_i$, and for all $s \prefix f$ we have that $\sup_{g \in \Sigma^\omega} \varPhi_i(sg) \geq v_i$.
	Hence, defining $v = \max(v_1, v_2)$ implies $\varPhi(f) \not\geq v$ as well as $\sup_{g \in \Sigma^\omega} \varPhi(sg) \geq v$ for all $s \prefix f$.
\qed\end{proof}

Our definition of liveness generalizes the boolean one.
A boolean property $P \subseteq \Sigma^\omega$ is live according to the classical definition \cite{DBLP:journals/ipl/AlpernS85} iff its characteristic property $\varPhi_P$ is live according to our definition.
Moreover, it is worth emphasizing that boolean liveness enjoys a stronger closure under maximum, namely, the union of a boolean liveness property and any property is live in the boolean sense.


However, as in the boolean setting, the intersection of safety and liveness contains only the single degenerate property that always outputs $\top$.

\begin{proposition}\label{prop:top}
	A property $\varPhi$ is safe and live iff $\varPhi(f) = \top$ for all $f \in \Sigma^\omega$.
\end{proposition}
\begin{proof}
	Observe that $\varPhi_{\top}$ is trivially safe and live.
	Now, let $\varPsi$ be a property that is both safe and live, and suppose towards contradiction that $\varPsi(f) < \top$ for some $f \in \Sigma^\omega$.
	Since $\varPsi$ is live,  there exists $v > \varPsi(f)$ such that for all $s \prec f$, we have $\sup_{g \in \Sigma^\omega} \varPsi(sg) \geq v$.
	In particular, $\inf_{s\prefix f}\sup_{g \in \Sigma^\omega} \varPsi(sg) \geq v > \varPsi(f)$ holds, implying $\varPsi^*(f) > \varPsi(f)$ by definition of safety closure.
	By Theorem~\ref{thm:safe:main}, this contradicts the assumption that $\varPsi$ is safe.
\qed\end{proof}

We define co-liveness symmetrically, and note that the duals of the observations above also hold for co-liveness.

\begin{definition}[Co-liveness]
	A property $\varPhi : \Sigma^\omega \to \D$ is \emph{co-live} iff for all $f \in \Sigma^\omega$, if $\varPhi(f) > \bot$, then there exists a value $v \in \D$ such that $\varPhi(f) \not \leq v$ and for all prefixes $s \prec f$, we have $\inf_{g \in \Sigma^\omega} \varPhi(sg) \leq v$.
\end{definition}

Next, we present some examples of liveness and co-liveness properties.
We start by showing that $\liminf$- and $\limsup$-properties can be live and co-live.

\begin{example}
	Let $\Sigma = \{a,b\}$ be an alphabet, and let $P = \LTLg \LTLf a$ and $Q = \LTLf \LTLg b$ be boolean properties defined in linear temporal logic.
	Consider their characteristic properties $\varPhi_P$ and $\varPhi_Q$.
	As we pointed out earlier, our definitions generalize their boolean counterparts, therefore $\varPhi_P$ and $\varPhi_Q$ are both live and co-live.
	Moreover, $\varPhi_P$ is a $\limsup$-property: define $\pi_P(s) = 1$ if $s \in \Sigma^* a$, and $\pi_P(s) = 0$ otherwise, and observe that $\varPhi_P(f) = \limsup_{s \prec f} \pi_P(s)$ for all $f \in \Sigma^\omega$.
	Similarly, $\varPhi_Q$ is a $\liminf$-property.
\qed\end{example}

Now, we show that the maximal response-time property is live, and the minimal response time is co-live.

\begin{example}
	Recall the co-safety property $\varPhi_{\max}$ of maximal response time from Example~\ref{ex:maxresp}.
	Let $f \in \Sigma^\omega$ such that $\varPhi_{\max}(f) < \infty$.
	We can extend every prefix $s \prec f$ with $g = \req\, \tick^\omega$, which gives us $\varPhi_{\max}(sg) = \infty > \varPhi(f)$.
	Equivalently, for every $f \in \Sigma^\omega$, we have $\varPhi_{\max}^*(f) = \infty > \varPhi_{\max}(f)$.
	Hence $\varPhi_{\max}$ is live and, analogously, the safety property $\varPhi_{\min}$ from Example~\ref{ex:minresp} is co-live.
\qed\end{example}

Finally, we show that the \emph{average response-time} property is live and co-live.

\begin{example} \label{ex:avgresp}
	Let $\Sigma = \{\req, \gra, \tick, \other\}$.
	For all $s \in \Sigma^*$, let $p(s) = 1$ if there is no pending $\req$ in $s$, and $p(s) = 0$ otherwise.
	Define $\pi_{\text{valid}}(s) = |\{r \preceq s \st \exists t \in \Sigma^* : r = t \, \req \land p(t) = 1\}|$ as the number of valid requests in $s$, and define $\pi_{\text{time}}(s)$ as the number of $\tick$ observations that occur after a valid $\req$ and before the matching $\gra$.
	Then, $\varPhi_{\avg} = (\pi_{\avg}, \liminf)$, where $\pi_{\avg}(s) = \frac{\pi_{\text{time}}(s)}{\pi_{\text{valid}}(s)}$ for all $s \in \Sigma^*$ with $\pi_{\text{valid}}(s) > 0$, and $\pi_{\avg}(s) = \infty$ otherwise.
	For example, $\pi_{\avg}(s) = \frac{3}{2}$ for $s = \req\, \tick\, \gra\, \tick\, \req\, \tick\, \req\, \tick$.
	Note that $\varPhi_{\avg}$ is a $\liminf$-property.
	
	The property $\varPhi_{\avg}$ is defined on the value domain $[0,\infty]$ and is both live and co-live.	
	To see this, let $f \in \Sigma^\omega$ such that $0 < \varPhi_{\avg}(f) < \infty$ and, for every prefix $s \prec f$, consider $g = \req\, \tick^\omega$ and $h = \gra\, (\req\,\gra)^\omega$.
	Since $s g$ has a pending request followed by infinitely many clock ticks, we have $\varPhi_{\avg}(s g) = \infty$.
	Similarly, since $s h$ eventually has all new requests immediately granted, we get $\varPhi_{\avg}(s h) = 0$.
\qed\end{example}

\subsection{The Quantitative Safety-Liveness Decomposition}
A celebrated theorem states that every boolean property can be expressed as an intersection of a safety property and a liveness property~\cite{DBLP:journals/ipl/AlpernS85}. 
In this section, we prove the analogous result for the quantitative setting. 

\begin{example}
	Let $\Sigma = \{\req, \gra, \tick, \other\}$.
	Recall the maximal response-time property $\varPhi_{\max}$ from Example~\ref{ex:maxresp}, and the average response-time property $\varPhi_{\avg}$ from Example~\ref{ex:avgresp}.
	Let $n > 0$ be an integer and define a new property $\varPhi$ by $\varPhi(f) = \varPhi_{\avg}(f)$ if $\varPhi_{\max}(f) \leq n$, and $\varPhi(f) = 0$ otherwise.
	For the safety closure of $\varPhi$, we have $\varPhi^*(f) = n$ if $\varPhi_{\max}(f) \leq n$, and $\varPhi^*(f) = 0$ otherwise.
	Now, we further define $\varPsi(f) = \varPhi_{\avg}(f)$ if $\varPhi_{\max}(f) \leq n$, and $\varPsi(f) = n$ otherwise.
	Observe that $\varPsi$ is live, because every prefix of a trace whose value is less than $n$ can be extended to a greater value.
	Finally, note that for all $f \in \Sigma^\omega$, we can express $\varPhi(f)$ as the pointwise minimum of $\varPhi^*(f)$ and $\varPsi(f)$.
	Intuitively, the safety part $\varPhi^*$ of this decomposition checks whether the maximal response time stays below the permitted bound, and the liveness part $\varPsi$ keeps track of the average response time as long as the bound is satisfied.
\qed\end{example}

Following a similar construction, we show that a safety-liveness decomposition exists for every property.

\begin{theorem}\label{thm:decomp}
	For every property $\varPhi$, there exists a liveness property $\varPsi$ such that $\varPhi(f) = \min(\varPhi^*(f), \varPsi(f))$ for all $f \in \Sigma^\omega$.
\end{theorem}
\begin{proof}
	Let $\varPhi$ be a property and consider its safety closure $\varPhi^*$.
	We define $\varPsi$ as follows:
	$\varPsi(f) = \varPhi(f)$ if $\varPhi^*(f) \neq \varPhi(f)$, and $\varPsi(f) = \top$ otherwise.
	Note that $\varPhi^*(f) \geq \varPhi(f)$ for all $f\in\Sigma^\omega$ by Proposition~\ref{proposition:safe:closure}.
	When $\varPhi^*(f) > \varPhi(f)$, we have $\min(\varPhi^*(f), \varPsi(f)) = \min(\varPhi^*(f), \varPhi(f)) = \varPhi(f)$.
	When $\varPhi^*(f) = \varPhi(f)$, we have $\min(\varPhi^*(f), \varPsi(f)) = \min(\varPhi(f), \top) = \varPhi(f)$.
	
	Now, suppose towards contradiction that $\varPsi$ is not live, i.e., there exists $f \in \Sigma^\omega$ such that $\varPsi(f) < \top$ and for all $v \not \leq \varPhi(f)$, there exists $s \prec f$ satisfying $\sup_{g \in \Sigma^\omega} \varPhi(sg) \not\geq v$.
	Let $f \in \Sigma^\omega$ be such that $\varPsi(f) < \top$.
	Then, by definition of $\varPsi$, we know that $\varPsi(f) = \varPhi(f) < \varPhi^*(f)$.
	Moreover, since $\varPhi^*(f) \not \leq \varPsi(f)$, there exists $s \prec f$ satisfying $\sup_{g \in \Sigma^\omega} \varPhi(sg) \not\geq \varPhi^*(f)$.
	In particular, we have $\sup_{g \in \Sigma^\omega} \varPhi(sg) < \varPhi^*(f)$, which is a contradiction since we have $\varPhi^*(f) = \inf_{r \prec f} \sup_{g \in \Sigma^\omega} \varPhi(rg)$ by definition, and $s \prec f$.
	Therefore, $\varPsi$ is live.
\qed\end{proof}

In particular, if the given property is safe or live, the decomposition is trivial.

\begin{remark}\label{rem:trivialdecomp}
	Let $\varPhi$ be a property.
	If $\varPhi$ is safe (resp.\ live), then the safety (resp.\ liveness) part of the decomposition is $\varPhi$ itself, and the liveness (resp.\ safety) part is the constant property that maps every trace to $\top$.
\end{remark}

Another, lesser-known decomposition theorem is the one of nonunary boolean properties into two liveness properties~\cite{DBLP:journals/ipl/AlpernS85}.
We extend this result to the quantitative setting.

\begin{theorem}\label{thm:livedecomp}
	For every property $\varPhi$, there exist two liveness properties $\varPsi_1$ and $\varPsi_2$ such that $\varPhi(f) = \min(\varPsi_1(f), \varPsi_2(f))$ for all $f \in \Sigma^\omega$.
\end{theorem}
\begin{proof}
	Let $\Sigma$ be a finite alphabet and $a_1, a_2 \in \Sigma$ be two distinct letters.
	Consider an arbitrary property $\varPhi$.
	For $i \in \{1, 2\}$, we define $\varPsi_i$ as follows:
	$\varPsi_i(f) = \top$ if $f = s(a_i)^\omega$ for some $s\in\Sigma^*$, and $\varPsi_i(f) = \varPhi(f)$ otherwise.
	Note that, since $a_1$ and $a_2$ are distinct, whenever $f \in \Sigma^* (a_1)^\omega$ then $f \notin \Sigma^* (a_2)^\omega$, and vice versa.
	Then, we have that both $\varPsi_1$ and $\varPsi_2$ are $\top$ only when $\varPhi$ is $\top$.
	In the remaining cases, when at most one of $\varPsi_1$ and $\varPsi_2$ is $\top$, then either both equals $\varPhi$ or one of them is $\top$ and the other is $\varPhi$.
	As a direct consequence, $\varPhi(f) = \min(\varPsi_1(f), \varPsi_2(f))$ for all $f\in\Sigma^\omega$.
	
	Now, we show that $\varPsi_1$ and $\varPsi_2$ are both live.
	By construction, $\varPsi_i(s(a_i)^\omega)=\top$ for all $s\in\Sigma^*$.
	In particular, $\varPsi^*_i(f) = \inf_{s \prefix f}\sup_{g \in \Sigma^\omega}\varPsi_i(sg) = \top$ for all $f\in\Sigma^\omega$.
	We conclude that $\varPsi_i$ is live thanks to Theorem~\ref{thm:livenessclosure}.
\qed\end{proof}

For co-safety and co-liveness, the duals of Remark~\ref{rem:trivialdecomp} and Theorems~\ref{thm:decomp}~and~\ref{thm:livedecomp}  hold.
In particular, every property is the pointwise maximum of its co-safety closure and a co-liveness property.

\subsection{Related Definitions of Quantitative Liveness}


In~\cite{DBLP:journals/isci/LiDL17}, the authors define a property $\varPhi$ as \emph{multi-live} iff $\varPhi^*(f) > \bot$ for all $f\in \Sigma^\omega$.
We show that our definition is more restrictive, resulting in fewer liveness properties while still allowing a safety-liveness decomposition.

\begin{proposition}\label{prop:multi:live}
	Every live property is multi-live, and the inclusion is strict.
\end{proposition}
\begin{proof}
	We prove that liveness implies multi-liveness.
	Suppose toward contradiction that some property $\varPhi$ is live, but not multi-live.
	Then, there exists $f\in\Sigma^\omega$ for which $\varPhi^*(f) = \bot$, and therefore $\varPhi(f) = \bot$ too.
	Note that we assume $\D$ is a nontrivial complete lattice, i.e., $\top \neq \bot$.
	Then, since $\varPhi$ is live, we have $\varPhi^*(f) > \varPhi(f)$ by Theorem~\ref{thm:livenessclosure}, which yields a contradiction.
\qed\end{proof}

We provide a separating example on a totally ordered domain below.

\begin{example}
	Let $\Sigma = \{a,b, c\}$, and consider the following property:
	$\varPhi(f) = 0$ if $f \models \LTLg a$, and $\varPhi(f) = 1$ if $f \models \LTLf c$, and $\varPhi(f) = 2$ otherwise (i.e., if $f \models \LTLf b \land \LTLg \lnot c$).
	For all $f\in\Sigma^\omega$ and prefixes $s\prefix f$, we have $\varPhi(s c^\omega) = 1$.
	Thus $\varPhi^*(f) \neq \bot$, which implies that $\varPhi$ is multi-live.
	However, $\varPhi$ is not live.
	Indeed, for every $f\in\Sigma^\omega$ such that $f \models \LTLf c$, we have $\varPhi(f) = 1 < \top$.
	Moreover, $f$ admits some prefix $s$ that contains an occurrence of $c$, thus satisfying $\sup_{g \in \Sigma^\omega} \varPhi(sg) = 1$.
\qed\end{example}

In~\cite{DBLP:conf/nfm/GorostiagaS22}, the authors define a property $\varPhi$ as \emph{verdict-live} iff for every $f \in \Sigma^\omega$ and value $v \not \leq \varPhi(f)$, every prefix $s \prec f$ satisfies $\varPhi(sg) = v$ for some $g \in \Sigma^\omega$.
We show that our definition is more liberal.

\begin{proposition}
	Every verdict-live property is live, and the inclusion is strict.
\end{proposition}

We provide a separating example below, concluding that our definition is strictly more general even for totally ordered domains.

\begin{example}
	Let $\Sigma = \{a,b\}$, and consider the following property:
	$\varPhi(f) = 0$ if $f \not \models \LTLf b$, and $\varPhi(f) = 1$ if $f \models \LTLf (b \land \LTLnext \LTLf b)$, and $\varPhi(f) = 2^{-|s|}$ otherwise, where $s \prec f$ is the shortest prefix in which $b$ occurs.
	Consider an arbitrary $f \in \Sigma^\omega$.
	If $\varPhi(f) = 1$, then the liveness condition is vacuously satisfied.
	If $\varPhi(f) = 0$, then $f = a^\omega$, and every prefix $s \prec f$ can be extended with $g = ba^\omega$ or $h = b^\omega$ to obtain $\varPhi(s g) = 2^{-(|s|+1)}$ and $\varPhi(s h) = 1$.
	If $0 < \varPhi(f) < 1$, then $f$ satisfies $\LTLf b$ but not $\LTLf (b \land \LTLnext \LTLf b)$, and every prefix $s \prec f$ can be extended with $b^\omega$ to obtain $\varPhi(s b^\omega) = 1$.
	Hence $\varPhi$ is live.
	However, $\varPhi$ is not verdict-live.
	To see this, consider the trace $f = a^k b a^\omega$ for some integer $k \geq 1$ and note that $\varPhi(f) = 2^{-(k+1)}$.
	Although all prefixes of $f$ can be extended to reach the value 1, the value domain contains elements between $\varPhi(f)$ and 1, namely the values $2^{-m}$ for $1 \leq m \leq k$.
	Each of these values can be rejected after reading a finite prefix of $f$, because for $n \geq m$ it is not possible to extend $a^n$ to reach the value $2^{-m}$. 
\qed\end{example}

\section{Approximate Monitoring through Approximate Safety}
In this section, we consider properties on extended reals $\R^{\pm \infty} = \R \cup \{-\infty, +\infty\}$.
We denote by $\R_{\geq 0}$ the set of nonnegative real numbers.

\begin{definition}[Approximate safety and co-safety]
	Let $\alpha \in \R_{\geq 0}$.
	A property $\varPhi$ is \emph{$\alpha$-safe} iff for every $f \in \Sigma^\omega$ and value $v \in \R^{\pm \infty}$ with $\varPhi(f) < v$, there exists a prefix $s \prefix f$ such that $\sup_{g \in \Sigma^\omega} \varPhi(sg) < v + \alpha$.
	Similarly, $\varPhi$ is \emph{$\alpha$-co-safe} iff for every $f \in \Sigma^\omega$ and $v \in \R^{\pm \infty}$ with $\varPhi(f) > v$, there exists $s \prefix f$ such that $\inf_{g \in \Sigma^\omega} \varPhi(sg) > v - \alpha$.
	When $\varPhi$ is $\alpha$-safe (resp.\ $\alpha$-co-safe) for some $\alpha \in \R_{\geq 0}$, we say that $\varPhi$ is \emph{approximately safe} (resp.\ \emph{approximately co-safe}).
\end{definition}

Approximate safety can be characterized through the following relation with the safety closure.

\begin{proposition}\label{prop:approx:safe}
  For every error bound $\alpha \in \R_{\geq 0}$, a property $\varPhi$ is $\alpha$-safe iff $\varPhi^*(f) - \varPhi(f) \leq \alpha$ for all $f \in \Sigma^\omega$.
\end{proposition}
\begin{proof}
	Let $\varPhi$ and $\alpha$ be as above.
	We show each direction separately by contradiction.
	First, assume $\varPhi$ is $\alpha$-safe.
	Suppose towards contradiction that $\varPhi^*(f) - \varPhi(f) > \alpha$ for some $f \in \Sigma^\omega$.
	Let $v = \varPhi^*(f) - \alpha$ and notice that, since $\varPhi$ is $\alpha$-safe, there exists $s \prec f$ such that $\sup_{g \in \Sigma^\omega} \varPhi(sg) < v + \alpha = \varPhi^*(f)$.
	By definition, we get $\sup_{g \in \Sigma^\omega} \varPhi(sg) < \inf_{r \prec f} \sup_{g \in \Sigma^\omega} \varPhi(rg)$, which is a contradiction.
	
	Now, assume $\varPhi^*(f) - \varPhi(f) \leq \alpha$ for all $f \in \Sigma^\omega$.
	Suppose towards contradiction that $\varPhi$ is not $\alpha$-safe, i.e., there exists $f \in \Sigma^\omega$ and $v \in \D$ such that (i) $\varPhi(f) < v$ and (ii) $\sup_{g \in \Sigma^\omega} \varPhi(sg) \geq v + \alpha$ for all $s \prec f$.
	Note that (i) implies $v + \alpha > \varPhi(f) + \alpha$, and (ii) implies $\inf_{s \prec f} \sup_{g \in \Sigma^\omega} \varPhi(sg) \geq v + \alpha$.
	Combining the two with the definition of $\varPhi^*$ we get $\varPhi^*(f) > \varPhi(f) + \alpha$, which is a contradiction.
\qed\end{proof}

An analogue of Proposition~\ref{prop:approx:safe} holds for approximate co-safety and the co-safety closure.
Moreover, approximate safety and approximate co-safety are dual notions that are connected by the complement operation, similarly to their precise counterparts (Theorem~\ref{thm:cosafe}).

\subsection{The Intersection of Approximate Safety and Co-safety}
Recall the ghost monitor from the introduction.
If, after a finite number of observations, all the possible prediction values are close enough, then we can simply freeze the current value and achieve a sufficiently small error.
This happens for properties that are both approximately safe and approximately co-safe, generalizing the unfolding approximation of discounted properties \cite{DBLP:journals/corr/BokerH14}.

\begin{proposition}\label{prop:squeeze}
  For every limit property $\varPhi$ and all error bounds $\alpha, \beta \in \R_{\geq 0}$, if $\varPhi$ is $\alpha$-safe and $\beta$-co-safe, then the set $\textstyle S_\delta = \{s \in \Sigma^* \st \sup_{r_1\in\Sigma^*} \varPhi(sr_1) - \inf_{r_2\in\Sigma^*} \varPhi(sr_2) \geq \delta\}$ is finite for all reals $\delta > \alpha + \beta$.
\end{proposition}
\begin{proof}
	Let $\alpha, \beta \in \R_{\geq 0}$ and $\varPhi$ be a limit property that is $\alpha$-safe and $\beta$-co-safe.
	Assume towards contradiction that $|S_\delta| = \infty$ for some $\delta > \alpha + \beta$.
	Notice that $S_\delta$ is prefix closed, i.e., for all $s, r \in \Sigma^*$ having both $r \prefixeq s$ and $s \in S_\delta$ implies $r \in S_\delta$.
	Then, by K\"onig's lemma, there exists $f \in \Sigma^\omega$ such that $s \in S_\delta$ for every prefix $s \prec f$.
	Let $s_i \prefix f$ be the prefix of length $i$.
	We have that $\lim_{n \to \infty} ( \sup_{r_1\in\Sigma^*} \varPhi(s_n r_1) - \inf_{r_2\in\Sigma^*} \varPhi(s_n r_2) ) \geq \delta > \alpha + \beta$.
	This implies that $\varPhi^*(f) - \varPhi_*(f) > \alpha + \beta$, which contradicts the assumption that $\varPhi$ is $\alpha$-safe and $\beta$-co-safe.
	Hence $S_\delta$ is finite for all $\delta > \alpha + \beta$.
\qed\end{proof}

Based on this proposition, we show that for limit properties that are both approximately safe and approximately co-safe, the influence of the suffix on the property value is eventually negligible.

\begin{theorem}\label{thm:prefix_determinacy}
  For every limit property $\varPhi$ such that $\varPhi(f) \in \R$ for all $f \in \Sigma^\omega$, and for all error bounds $\alpha, \beta \in \R_{\geq 0}$, if $\varPhi$ is $\alpha$-safe and $\beta$-co-safe, then for every real $\delta > \alpha + \beta$ and trace $f \in \Sigma^\omega$, there is a prefix $s \prec f$ such that for all continuations $w \in \Sigma^* \cup \Sigma^\omega$, we have $|\varPhi(sw) - \varPhi(s)| < \delta$.
\end{theorem}
\begin{proof}
	Given $\alpha,\beta \in\R_{\geq 0}$ and $\varPhi$ as in the statement, assume $\varPhi$ is $\alpha$-safe and $\beta$-co-safe.
	Let $\delta > \alpha + \beta$ and $f \in \Sigma^\omega$ be arbitrary.
	Let $S_\delta$ be as in Proposition~\ref{prop:squeeze}.
	Since $S_\delta$ is finite and prefix closed, there exists $s \prec f$ such that $sr \notin S_\delta$ for all $r \in \Sigma^*$.
	Let $s \prec f$ be the shortest such prefix.
	By construction, $\sup_{r_1\in\Sigma^*} \varPhi(s r_1) - \inf_{r_2\in\Sigma^*} \varPhi(s r_2) < \delta$.
	Furthermore, for all $t\in\Sigma^*$, we trivially have $\inf_{r_2\in\Sigma^*} \varPhi(s r_2) \leq \varPhi(s t) \leq \sup_{r_1\in\Sigma^*} \varPhi(s r_1)$.
	In particular, $\inf_{r_2\in\Sigma^*} \varPhi(s r_2) \leq \varPhi(s) \leq \sup_{r_1\in\Sigma^*} \varPhi(s r_1)$ holds simply by taking $t = \varepsilon$.
	Then, one can easily obtain $-\delta < \varPhi(s r) - \varPhi(s) < \delta$ for all $r \in \Sigma^*$.
	Since $\varPhi$ is a limit property, this implies $-\delta < \varPhi(s g) - \varPhi(s) < \delta$ for all $g \in \Sigma^*$ as well.
\qed\end{proof}

We illustrate this theorem with a \emph{discounted safety} property.

\begin{example} \label{ex:discsafe}
	Let $P \subseteq \Sigma^\omega$ be a boolean safety property.
	We define the finitary property $\pi_P : \Sigma^* \to [0,1]$ as follows: 
	$\pi_P(s) = 1$ if $sf \in P$ for some $f \in \Sigma^\omega$, and $\pi_P(s) = 1 - 2^{-|r|}$ otherwise, where $r \preceq s$ is the shortest prefix with $rf \notin P$ for all $f \in \Sigma^\omega$.
	The limit property $\varPhi = (\pi_P, \inf)$ is called \emph{discounted safety} \cite{DBLP:conf/icalp/AlfaroHM03}.
	Because $\varPhi$ is an $\inf$-property, it is safe by Theorem~\ref{thm:safe:inf}.
	Now consider the finitary property $\pi_P'$ defined by $\pi_P'(s) = 1 - 2^{-|s|}$ if $sf \in P$ for some $f \in \Sigma^\omega$, and $\pi_P'(s) = 1 - 2^{-|r|}$ otherwise, where $r \preceq s$ is the shortest prefix with $rf \notin P$ for all $f \in \Sigma^\omega$.
	Let $\varPhi' = (\pi_P', \sup)$, and note that $\varPhi(f) = \varPhi'(f)$ for all $f \in \Sigma^\omega$.
	Hence $\varPhi$ is also co-safe, because it is a $\sup$-property.
        
	Let $f \in \Sigma^\omega$ and $\delta > 0$.
	For every prefix $s \prec f$, the set of possible prediction values is either the range $[1 - 2^{-|s|}, 1]$ or the singleton $\{1 - 2^{-|r|}\}$, where $r \preceq s$ is chosen as above.
	In the latter case, we have $|\varPhi(sw) - \varPhi(s)| = 0 < \delta$ for all $w \in \Sigma^* \cup \Sigma^\omega$.
	In the former case, since the range becomes smaller as the prefix grows, there is a prefix $s' \prec f$ with $2^{-|s'|} < \delta$, which yields $|\varPhi(s'w) - \varPhi(s')| < \delta$ for all $w \in \Sigma^* \cup \Sigma^\omega$.
\qed\end{example}

\subsection{Finite-state Approximate Monitoring}
Monitors with finite state spaces are particularly desirable, because finite automata enjoy a plethora of desirable closure and decidability properties.
Here, we prove that properties that are both approximately safe and approximately co-safe can be monitored approximately by a finite-state monitor.
First, we recall the notion of abstract quantitative monitor from~\cite{DBLP:conf/rv/HenzingerMS22}.

A binary relation ${\sim}$ over $\Sigma^*$ is an \emph{equivalence relation} iff it is reflexive, symmetric, and transitive.
Such a relation is \emph{right-monotonic} iff  $s_1 \sim s_2$ implies $s_1 r \sim s_2 r$ for all $s_1,s_2,r \in \Sigma^*$.
For an equivalence relation ${\sim}$ over $\Sigma^*$ and a finite trace $s \in \Sigma^*$, we write $[s]_{{\sim}}$ for the equivalence class of ${\sim}$ to which $s$ belongs.
When ${\sim}$ is clear from the context, we write $[s]$ instead.
We denote by $\Sigma^* / {\sim}$ the quotient of the relation ${\sim}$.

\begin{definition}[Abstract monitor~\cite{DBLP:conf/rv/HenzingerMS22}]
  An \emph{abstract monitor} $\calM = (\sim, \gamma)$ is a pair consisting of a right-monotonic equivalence relation ${\sim}$ on $\Sigma^*$ and a function $\gamma \colon ({\Sigma^* / \sim}) \rightarrow \RR$.
  The monitor $\calM$ is \emph{finite-state} iff the relation $\sim$ has finitely many equivalence classes.
	Let $\delta_{\fin},\delta_{\lim} \in \RR$ be error bounds.
	We say that $\calM$ is a \emph{$(\delta_{\fin},\delta_{\lim})$-monitor} for a given limit property $\varPhi = (\pi, \ell)$ iff for all $s \in \Sigma^*$ and $f \in \Sigma^\omega$, we have $|\pi(s) - \gamma([s])| \leq \delta_{\fin}$ and $|\ell_{s \prefix f}(\pi(s)) - \ell_{s\prefix f}(\gamma([s]))| \leq \delta_{\lim}$.
\end{definition}

Building on Theorem~\ref{thm:prefix_determinacy}, we identify a sufficient condition to guarantee the existence of an abstract monitor with finitely many equivalence classes.

\begin{theorem} \label{thm:fin_apx_mon}
  For every limit property $\varPhi$ such that $\varPhi(f) \in \R$ for all $f \in \Sigma^\omega$, and for all error bounds $\alpha, \beta \in \R_{\geq 0}$, if $\varPhi$ is $\alpha$-safe and $\beta$-co-safe, then for every real $\delta > \alpha + \beta$, there exists a finite-state $(\delta,\delta)$-monitor for $\varPhi$.
\end{theorem}
\begin{proof}
	Let $\alpha, \beta \in \R_{\geq 0}$, and $\varPhi$ be a limit property such that $\varPhi(f) \in \R$ for all $f \in \Sigma^\omega$.
	Assume $\varPhi$ is $\alpha$-safe and $\beta$-co-safe, and let $\delta > \alpha + \beta$.
	We show how to construct a finite-state $(\delta,\delta)$-monitor for $\varPhi$.
	
	Consider the finite set $S_\delta$ from Proposition~\ref{prop:squeeze}.
	If $S_\delta$ is empty, then $|\varPhi(s_1) - \varPhi(s_2)| \leq \delta$ holds for all $s_1, s_2 \in \Sigma^*$, and thus we can construct a trivial $(\delta, \delta)$-monitor for $\varPhi$ simply by (arbitrarily) mapping all finite traces to $\varPhi(\varepsilon)$.
	So, we assume without loss of generality that $S_\delta$ is not empty.
	
	Consider the function ${\prefixeq_{S_\delta}} \colon \Sigma^*\rightarrow \Sigma^*$ such that ${\prefixeq_{S_\delta}}(s) = s$ if $s \in S_\delta$, and ${\prefixeq_{S_\delta}}(s) = s'$ otherwise, where $s' \preceq s$ is the the shortest prefix with $s' \notin S_\delta$.
	We let $\calM = ({\sim}, \gamma)$ where ${\sim} = \{(s_1, s_2) \st {\prefixeq_{S_\delta}}(s_1) = {\prefixeq_{S_\delta}}(s_2)\}$ and $\gamma([s]) = \varPhi({\prefixeq_{S_\delta}}(s))$.
	By construction, ${\sim}$ is right-monotonic and has at most $2|S_\delta|$ equivalence classes.
	
	Now, we prove that $|\varPhi(s) - \gamma([s])| \leq \delta$ for all $s\in\Sigma^*$.
	If $s \in S_\delta$, then $\gamma([s]) = \varPhi(s)$ by definition, and the statement holds trivially.
	Otherwise, if $s \notin S_\delta$, we let $r = {\prefixeq_{S_\delta}}(s)$, which gives us $|\varPhi(r t_1) - \varPhi(r t_2)| < \delta$ for all $t_1, t_2 \in \Sigma^*$.
	In particular, $|\varPhi(s) - \gamma([s])| < \delta$ since $r \preceq s$.
	We remark that an error of at most $\delta$ on finite traces implies an error of at most $\delta$ on infinite traces.
	
	Finally, we prove that ${\sim}$ is right-monotonic.
	Let $s_1,s_2 \in \Sigma^*$ such that $s_1 \sim s_2$.
	Note that $s_1 \sim s_2$ implies $s_1 \in S_\delta \Leftrightarrow s_2 \in S_\delta$ by definition of ${\prefixeq_{S_\delta}}$.
	If $s_1, s_2\in S_\delta$, then ${\prefixeq_{S_\delta}}$ is the identity function, and thus $s_1t \sim s_2t$ for all $t\in\Sigma^*$ trivially.
	Otherwise, if $s_1, s_2\notin S_\delta$, we define $s = {\prefixeq_{S_\delta}}(s_1) = {\prefixeq_{S_\delta}}(s_2)\notin S_\delta$.
	By the definition of ${\prefixeq_{S_\delta}}$, we have that ${\prefixeq_{S_\delta}}(s) \notin S_\delta$ implies ${\prefixeq_{S_\delta}}(st) = {\prefixeq_{S_\delta}}(s)$ for all $t \in \Sigma^*$.
	In particular, $s_1t \sim s_2t$.
\qed\end{proof}

Due to Theorem~\ref{thm:fin_apx_mon}, the discounted safety property of Example~\ref{ex:discsafe} has a finite-state monitor for every positive error bound.
We remark that Theorem~\ref{thm:fin_apx_mon} is proved by a construction that generalizes the unfolding approach for the approximate determinization of discounted automata~\cite{DBLP:conf/fsttcs/BokerH12}, which unfolds an automaton until the distance constraint is satisfied.

\section{Conclusion}
We presented a generalization of safety and liveness that lifts the safety-progress hierarchy to the quantitative setting of~\cite{DBLP:journals/tocl/ChatterjeeDH10} while preserving major desirable features of the boolean setting, such as the safety-liveness decomposition.

Monitorability identifies a boundary separating properties that can be verified or falsified from a finite number of observations, from those that cannot.
Safety-liveness and co-safety-co-liveness decompositions allow us separate, for an individual property, monitorable parts from nonmonitorable parts.
The larger the monitorable parts of the given property, the stronger the decomposition.
We provided the strongest known safety-liveness decomposition, which consists of a pointwise minimum between a safe part defined by a quantitative safety closure, and a live part which corrects for the difference.
We then defined approximate safety as the relaxation of safety by a parametric error bound.
This further increases the monitorability of properties and offers monitorability at a parametric cost.
In fact, we showed that every property that is both approximately safe and approximately co-safe can be monitored arbitrarily precisely by a finite-state monitor.
A future direction is to extend our decomposition to approximate safety together with a support for quantitative assumptions~\cite{DBLP:conf/rv/HenzingerS20}.

The literature contains efficient model-checking procedures that leverage the boolean safety hypothesis~\cite{DBLP:journals/fmsd/KupfermanV01,DBLP:conf/spin/Latvala03}.
We thus expect that also quantitative safety and co-safety, and their approximations, enable efficient verification algorithms for quantitative properties.


\bibliographystyle{splncs04}
\bibliography{qsl_journal}

\begin{thebibliography}{10}
\providecommand{\url}[1]{\texttt{#1}}
\providecommand{\urlprefix}{URL }
\providecommand{\doi}[1]{https://doi.org/#1}

\bibitem{DBLP:journals/tcs/AlfaroFHMS05}
de~Alfaro, L., Faella, M., Henzinger, T.A., Majumdar, R., Stoelinga, M.: Model
  checking discounted temporal properties. Theor. Comput. Sci.
  \textbf{345}(1),  139--170 (2005). \doi{10.1016/j.tcs.2005.07.033}

\bibitem{DBLP:conf/icalp/AlfaroFS04}
de~Alfaro, L., Faella, M., Stoelinga, M.: Linear and branching metrics for
  quantitative transition systems. In: D{\'{\i}}az, J., Karhum{\"{a}}ki, J.,
  Lepist{\"{o}}, A., Sannella, D. (eds.) Automata, Languages and Programming:
  31st International Colloquium, {ICALP} 2004, Turku, Finland, July 12-16,
  2004. Proceedings. Lecture Notes in Computer Science, vol.~3142, pp. 97--109.
  Springer (2004). \doi{10.1007/978-3-540-27836-8\_11}

\bibitem{DBLP:conf/icalp/AlfaroHM03}
de~Alfaro, L., Henzinger, T.A., Majumdar, R.: Discounting the future in systems
  theory. In: Baeten, J.C.M., Lenstra, J.K., Parrow, J., Woeginger, G.J. (eds.)
  Automata, Languages and Programming, 30th International Colloquium, {ICALP}
  2003, Eindhoven, The Netherlands, June 30 - July 4, 2003. Proceedings.
  Lecture Notes in Computer Science, vol.~2719, pp. 1022--1037. Springer
  (2003). \doi{10.1007/3-540-45061-0\_79}

\bibitem{DBLP:journals/ipl/AlpernS85}
Alpern, B., Schneider, F.B.: Defining liveness. Inf. Process. Lett.
  \textbf{21}(4),  181--185 (1985). \doi{10.1016/0020-0190(85)90056-0}

\bibitem{DBLP:journals/dc/AlpernS87}
Alpern, B., Schneider, F.B.: Recognizing safety and liveness. Distributed
  Comput.  \textbf{2}(3),  117--126 (1987). \doi{10.1007/BF01782772}

\bibitem{DBLP:series/lncs/BartocciFFR18}
Bartocci, E., Falcone, Y., Francalanza, A., Reger, G.: Introduction to runtime
  verification. In: Bartocci, E., Falcone, Y. (eds.) Lectures on Runtime
  Verification - Introductory and Advanced Topics, Lecture Notes in Computer
  Science, vol. 10457, pp. 1--33. Springer (2018).
  \doi{10.1007/978-3-319-75632-5\_1}

\bibitem{DBLP:journals/logcom/BauerLS10}
Bauer, A., Leucker, M., Schallhart, C.: Comparing {LTL} semantics for runtime
  verification. J. Log. Comput.  \textbf{20}(3),  651--674 (2010).
  \doi{10.1093/logcom/exn075}

\bibitem{DBLP:journals/tosem/BauerLS11}
Bauer, A., Leucker, M., Schallhart, C.: Runtime verification for {LTL} and
  {TLTL}. {ACM} Trans. Softw. Eng. Methodol.  \textbf{20}(4),  14:1--14:64
  (2011). \doi{10.1145/2000799.2000800}

\bibitem{DBLP:conf/cav/BloemCHJ09}
Bloem, R., Chatterjee, K., Henzinger, T.A., Jobstmann, B.: Better quality in
  synthesis through quantitative objectives. In: Bouajjani, A., Maler, O.
  (eds.) Computer Aided Verification, 21st International Conference, {CAV}
  2009, Grenoble, France, June 26 - July 2, 2009. Proceedings. Lecture Notes in
  Computer Science, vol.~5643, pp. 140--156. Springer (2009).
  \doi{10.1007/978-3-642-02658-4\_14}

\bibitem{DBLP:reference/mc/BloemCJ18}
Bloem, R., Chatterjee, K., Jobstmann, B.: Graph games and reactive synthesis.
  In: Clarke, E.M., Henzinger, T.A., Veith, H., Bloem, R. (eds.) Handbook of
  Model Checking, pp. 921--962. Springer (2018).
  \doi{10.1007/978-3-319-10575-8\_27}

\bibitem{DBLP:journals/tocl/BokerCHK14}
Boker, U., Chatterjee, K., Henzinger, T.A., Kupferman, O.: Temporal
  specifications with accumulative values. {ACM} Trans. Comput. Log.
  \textbf{15}(4),  27:1--27:25 (2014). \doi{10.1145/2629686}

\bibitem{DBLP:conf/fsttcs/BokerH12}
Boker, U., Henzinger, T.A.: Approximate determinization of quantitative
  automata. In: D'Souza, D., Kavitha, T., Radhakrishnan, J. (eds.) {IARCS}
  Annual Conference on Foundations of Software Technology and Theoretical
  Computer Science, {FSTTCS} 2012, December 15-17, 2012, Hyderabad, India.
  LIPIcs, vol.~18, pp. 362--373. Schloss Dagstuhl - Leibniz-Zentrum f{\"{u}}r
  Informatik (2012). \doi{10.4230/LIPIcs.FSTTCS.2012.362}

\bibitem{DBLP:journals/corr/BokerH14}
Boker, U., Henzinger, T.A.: Exact and approximate determinization of
  discounted-sum automata. Log. Methods Comput. Sci.  \textbf{10}(1) (2014).
  \doi{10.2168/LMCS-10(1:10)2014}

\bibitem{DBLP:journals/cacm/BouyerFLM11}
Bouyer, P., Fahrenberg, U., Larsen, K.G., Markey, N.: Quantitative analysis of
  real-time systems using priced timed automata. Commun. {ACM}  \textbf{54}(9),
   78--87 (2011). \doi{10.1145/1995376.1995396}

\bibitem{DBLP:journals/acta/BouyerMRLL18}
Bouyer, P., Markey, N., Randour, M., Larsen, K.G., Laursen, S.: Average-energy
  games. Acta Informatica  \textbf{55}(2),  91--127 (2018).
  \doi{10.1007/s00236-016-0274-1}

\bibitem{DBLP:journals/tcs/CernyHR12}
Cern{\'{y}}, P., Henzinger, T.A., Radhakrishna, A.: Simulation distances.
  Theor. Comput. Sci.  \textbf{413}(1),  21--35 (2012).
  \doi{10.1016/j.tcs.2011.08.002}

\bibitem{ChangMP93}
Chang, E., Manna, Z., Pnueli, A.: The safety-progress classification. In:
  Bauer, F.L., Brauer, W., Schwichtenberg, H. (eds.) Logic and Algebra of
  Specification. pp. 143--202. Springer Berlin Heidelberg, Berlin, Heidelberg
  (1993). \doi{10.1007/978-3-642-58041-3_5}

\bibitem{DBLP:journals/tocl/ChatterjeeDH10}
Chatterjee, K., Doyen, L., Henzinger, T.A.: Quantitative languages. {ACM}
  Trans. Comput. Log.  \textbf{11}(4),  23:1--23:38 (2010).
  \doi{10.1145/1805950.1805953}

\bibitem{DBLP:journals/tocl/ChatterjeeHO17}
Chatterjee, K., Henzinger, T.A., Otop, J.: Nested weighted automata. {ACM}
  Trans. Comput. Log.  \textbf{18}(4),  31:1--31:44 (2017).
  \doi{10.1145/3152769}

\bibitem{DBLP:conf/cav/DAntoniSS16}
D'Antoni, L., Samanta, R., Singh, R.: Qlose: Program repair with quantitative
  objectives. In: Chaudhuri, S., Farzan, A. (eds.) Computer Aided Verification
  - 28th International Conference, {CAV} 2016, Toronto, ON, Canada, July 17-23,
  2016, Proceedings, Part {II}. Lecture Notes in Computer Science, vol.~9780,
  pp. 383--401. Springer (2016). \doi{10.1007/978-3-319-41540-6\_21}

\bibitem{DBLP:conf/aplas/FahrenbergL13}
Fahrenberg, U., Legay, A.: Generalized quantitative analysis of metric
  transition systems. In: Shan, C. (ed.) Programming Languages and Systems -
  11th Asian Symposium, {APLAS} 2013, Melbourne, VIC, Australia, December 9-11,
  2013. Proceedings. Lecture Notes in Computer Science, vol.~8301, pp.
  192--208. Springer (2013). \doi{10.1007/978-3-319-03542-0\_14}

\bibitem{DBLP:journals/tcs/FahrenbergL14}
Fahrenberg, U., Legay, A.: The quantitative linear-time-branching-time
  spectrum. Theor. Comput. Sci.  \textbf{538},  54--69 (2014).
  \doi{10.1016/j.tcs.2013.07.030}

\bibitem{DBLP:journals/sttt/FalconeFM12}
Falcone, Y., Fernandez, J., Mounier, L.: What can you verify and enforce at
  runtime? Int. J. Softw. Tools Technol. Transf.  \textbf{14}(3),  349--382
  (2012). \doi{10.1007/s10009-011-0196-8}

\bibitem{DBLP:journals/acta/FaranK18}
Faran, R., Kupferman, O.: Spanning the spectrum from safety to liveness. Acta
  Informatica  \textbf{55}(8),  703--732 (2018).
  \doi{10.1007/s00236-017-0307-4}

\bibitem{DBLP:conf/csl/FerrereHK20}
Ferr{\`{e}}re, T., Henzinger, T.A., Kragl, B.: Monitoring event frequencies.
  In: Fern{\'{a}}ndez, M., Muscholl, A. (eds.) 28th {EACSL} Annual Conference
  on Computer Science Logic, {CSL} 2020, January 13-16, 2020, Barcelona, Spain.
  LIPIcs, vol.~152, pp. 20:1--20:16. Schloss Dagstuhl - Leibniz-Zentrum
  f{\"{u}}r Informatik (2020). \doi{10.4230/LIPIcs.CSL.2020.20}

\bibitem{DBLP:conf/lics/FerrereHS18}
Ferr{\`{e}}re, T., Henzinger, T.A., Sara{\c{c}}, N.E.: A theory of register
  monitors. In: Dawar, A., Gr{\"{a}}del, E. (eds.) Proceedings of the 33rd
  Annual {ACM/IEEE} Symposium on Logic in Computer Science, {LICS} 2018,
  Oxford, UK, July 09-12, 2018. pp. 394--403. {ACM} (2018).
  \doi{10.1145/3209108.3209194}

\bibitem{DBLP:conf/nfm/GorostiagaS22}
Gorostiaga, F., S{\'{a}}nchez, C.: Monitorability of expressive verdicts. In:
  Deshmukh, J.V., Havelund, K., Perez, I. (eds.) {NASA} Formal Methods - 14th
  International Symposium, {NFM} 2022, Pasadena, CA, USA, May 24-27, 2022,
  Proceedings. Lecture Notes in Computer Science, vol. 13260, pp. 693--712.
  Springer (2022). \doi{10.1007/978-3-031-06773-0\_37}

\bibitem{DBLP:conf/tacas/HavelundR02}
Havelund, K., Rosu, G.: Synthesizing monitors for safety properties. In:
  Katoen, J., Stevens, P. (eds.) Tools and Algorithms for the Construction and
  Analysis of Systems, 8th International Conference, {TACAS} 2002, Held as Part
  of the Joint European Conference on Theory and Practice of Software, {ETAPS}
  2002, Grenoble, France, April 8-12, 2002, Proceedings. Lecture Notes in
  Computer Science, vol.~2280, pp. 342--356. Springer (2002).
  \doi{10.1007/3-540-46002-0\_24}

\bibitem{DBLP:journals/ife/Henzinger13}
Henzinger, T.A.: Quantitative reactive modeling and verification. Comput. Sci.
  Res. Dev.  \textbf{28}(4),  331--344 (2013). \doi{10.1007/s00450-013-0251-7}

\bibitem{DBLP:conf/rv/HenzingerMS22}
Henzinger, T.A., Mazzocchi, N., Sara{\c{c}}, N.E.: Abstract monitors for
  quantitative specifications. In: Dang, T., Stolz, V. (eds.) Runtime
  Verification - 22nd International Conference, {RV} 2022, Tbilisi, Georgia,
  September 28-30, 2022, Proceedings. Lecture Notes in Computer Science, vol.
  13498, pp. 200--220. Springer (2022). \doi{10.1007/978-3-031-17196-3\_11}

\bibitem{DBLP:conf/concur/HenzingerO13}
Henzinger, T.A., Otop, J.: From model checking to model measuring. In:
  D'Argenio, P.R., Melgratti, H.C. (eds.) {CONCUR} 2013 - Concurrency Theory -
  24th International Conference, {CONCUR} 2013, Buenos Aires, Argentina, August
  27-30, 2013. Proceedings. Lecture Notes in Computer Science, vol.~8052, pp.
  273--287. Springer (2013). \doi{10.1007/978-3-642-40184-8\_20}

\bibitem{DBLP:conf/rv/HenzingerS20}
Henzinger, T.A., Sara{\c{c}}, N.E.: Monitorability under assumptions. In:
  Deshmukh, J., Nickovic, D. (eds.) Runtime Verification - 20th International
  Conference, {RV} 2020, Los Angeles, CA, USA, October 6-9, 2020, Proceedings.
  Lecture Notes in Computer Science, vol. 12399, pp. 3--18. Springer (2020).
  \doi{10.1007/978-3-030-60508-7\_1}

\bibitem{DBLP:conf/lics/HenzingerS21}
Henzinger, T.A., Sara{\c{c}}, N.E.: Quantitative and approximate monitoring.
  In: 36th Annual {ACM/IEEE} Symposium on Logic in Computer Science, {LICS}
  2021, Rome, Italy, June 29 - July 2, 2021. pp. 1--14. {IEEE} (2021).
  \doi{10.1109/LICS52264.2021.9470547}

\bibitem{DBLP:conf/csl/KatoenSZ14}
Katoen, J., Song, L., Zhang, L.: Probably safe or live. In: Henzinger, T.A.,
  Miller, D. (eds.) Joint Meeting of the Twenty-Third {EACSL} Annual Conference
  on Computer Science Logic {(CSL)} and the Twenty-Ninth Annual {ACM/IEEE}
  Symposium on Logic in Computer Science (LICS), {CSL-LICS} '14, Vienna,
  Austria, July 14 - 18, 2014. pp. 55:1--55:10. {ACM} (2014).
  \doi{10.1145/2603088.2603147}

\bibitem{DBLP:journals/entcs/KimKLSV02}
Kim, M., Kannan, S., Lee, I., Sokolsky, O., Viswanathan, M.: Computational
  analysis of run-time monitoring - fundamentals of java-mac. In: Havelund, K.,
  Rosu, G. (eds.) Runtime Verification 2002, {RV} 2002, FLoC Satellite Event,
  Copenhagen, Denmark, July 26, 2002. Electronic Notes in Theoretical Computer
  Science, vol.~70, pp. 80--94. Elsevier (2002).
  \doi{10.1016/S1571-0661(04)80578-4}

\bibitem{DBLP:journals/fmsd/KupfermanV01}
Kupferman, O., Vardi, M.Y.: Model checking of safety properties. Formal Methods
  Syst. Des.  \textbf{19}(3),  291--314 (2001). \doi{10.1023/A:1011254632723}

\bibitem{KwiatkowskaNP18}
Kwiatkowska, M., Norman, G., Parker, D.: Probabilistic Model Checking: Advances
  and Applications, pp. 73--121. Springer International Publishing, Cham
  (2018). \doi{10.1007/978-3-319-57685-5_3}

\bibitem{DBLP:conf/sigsoft/Kwiatkowska07}
Kwiatkowska, M.Z.: Quantitative verification: models techniques and tools. In:
  Crnkovic, I., Bertolino, A. (eds.) Proceedings of the 6th joint meeting of
  the European Software Engineering Conference and the {ACM} {SIGSOFT}
  International Symposium on Foundations of Software Engineering, 2007,
  Dubrovnik, Croatia, September 3-7, 2007. pp. 449--458. {ACM} (2007).
  \doi{10.1145/1287624.1287688}

\bibitem{DBLP:journals/tse/Lamport77}
Lamport, L.: Proving the correctness of multiprocess programs. {IEEE} Trans.
  Software Eng.  \textbf{3}(2),  125--143 (1977). \doi{10.1109/TSE.1977.229904}

\bibitem{DBLP:conf/spin/Latvala03}
Latvala, T.: Efficient model checking of safety properties. In: Ball, T.,
  Rajamani, S.K. (eds.) Model Checking Software, 10th International {SPIN}
  Workshop. Portland, OR, USA, May 9-10, 2003, Proceedings. Lecture Notes in
  Computer Science, vol.~2648, pp. 74--88. Springer (2003).
  \doi{10.1007/3-540-44829-2\_5}

\bibitem{DBLP:journals/isci/LiDL17}
Li, Y., Droste, M., Lei, L.: Model checking of linear-time properties in
  multi-valued systems. Inf. Sci.  \textbf{377},  51--74 (2017).
  \doi{10.1016/j.ins.2016.10.030}

\bibitem{DBLP:journals/scp/MannaP84}
Manna, Z., Pnueli, A.: Adequate proof principles for invariance and liveness
  properties of concurrent programs. Sci. Comput. Program.  \textbf{4}(3),
  257--289 (1984). \doi{10.1016/0167-6423(84)90003-0}

\bibitem{DBLP:conf/birthday/PeledH18}
Peled, D., Havelund, K.: Refining the safety-liveness classification of
  temporal properties according to monitorability. In: Margaria, T., Graf, S.,
  Larsen, K.G. (eds.) Models, Mindsets, Meta: The What, the How, and the Why
  Not? - Essays Dedicated to Bernhard Steffen on the Occasion of His 60th
  Birthday. Lecture Notes in Computer Science, vol. 11200, pp. 218--234.
  Springer (2018). \doi{10.1007/978-3-030-22348-9\_14}

\bibitem{DBLP:conf/fm/PnueliZ06}
Pnueli, A., Zaks, A.: {PSL} model checking and run-time verification via
  testers. In: Misra, J., Nipkow, T., Sekerinski, E. (eds.) {FM} 2006: Formal
  Methods, 14th International Symposium on Formal Methods, Hamilton, Canada,
  August 21-27, 2006, Proceedings. Lecture Notes in Computer Science,
  vol.~4085, pp. 573--586. Springer (2006). \doi{10.1007/11813040\_38}

\bibitem{DBLP:journals/tr/QianSCP22}
Qian, J., Shi, F., Cai, Y., Pan, H.: Approximate safety properties in metric
  transition systems. {IEEE} Trans. Reliab.  \textbf{71}(1),  221--234 (2022).
  \doi{10.1109/TR.2021.3139616}

\bibitem{DBLP:journals/fac/Sistla94}
Sistla, A.P.: Safety, liveness and fairness in temporal logic. Formal Aspects
  Comput.  \textbf{6}(5),  495--512 (1994). \doi{10.1007/BF01211865}

\bibitem{DBLP:journals/jlp/ThraneFL10}
Thrane, C.R., Fahrenberg, U., Larsen, K.G.: Quantitative analysis of weighted
  transition systems. J. Log. Algebraic Methods Program.  \textbf{79}(7),
  689--703 (2010). \doi{10.1016/j.jlap.2010.07.010}

\bibitem{DBLP:conf/atva/WeinerHKPS13}
Weiner, S., Hasson, M., Kupferman, O., Pery, E., Shevach, Z.: Weighted safety.
  In: Hung, D.V., Ogawa, M. (eds.) Automated Technology for Verification and
  Analysis - 11th International Symposium, {ATVA} 2013, Hanoi, Vietnam, October
  15-18, 2013. Proceedings. Lecture Notes in Computer Science, vol.~8172, pp.
  133--147. Springer (2013). \doi{10.1007/978-3-319-02444-8\_11}

\end{thebibliography}

\end{document}